\documentclass[12pt]{article}
\usepackage[margin=1in]{geometry}
\usepackage{amsmath,amssymb,amsfonts,amsthm,mathtools}
\usepackage{booktabs}
\usepackage{array}
\usepackage{enumitem}
\usepackage{microtype}
\usepackage[colorlinks=true,linkcolor=blue,citecolor=blue,urlcolor=blue]{hyperref}

\newtheorem{theorem}{Theorem}
\newtheorem{lemma}{Lemma}
\newtheorem{corollary}{Corollary}
\newtheorem{proposition}{Proposition}
\theoremstyle{definition}
\newtheorem{assumption}{Assumption}
\newtheorem{definition}{Definition}
\theoremstyle{remark}
\newtheorem{remark}{Remark}
\newtheorem{example}{Example}
\newenvironment{examplecont}[1]{\par\medskip\noindent\textit{#1.}\enspace\ignorespaces}{\par\medskip}

\newcommand{\bY}{\mathbf{Y}}
\newcommand{\bZ}{\mathbf{Z}}
\newcommand{\bD}{\mathbf{D}}
\newcommand{\bDt}{\widetilde{\mathbf{D}}}
\newcommand{\bW}{\mathbf{W}}
\newcommand{\bX}{\mathbf{X}}
\newcommand{\bR}{\mathbf{R}}
\newcommand{\bA}{\mathbf{A}}
\newcommand{\bB}{\mathbf{B}}
\newcommand{\bC}{\mathbf{C}}
\newcommand{\bM}{\mathbf{M}}
\newcommand{\bV}{\mathbf{V}}
\newcommand{\bG}{\mathbf{G}}
\newcommand{\bH}{\mathbf{H}}
\newcommand{\bI}{\mathbf{I}}
\newcommand{\ba}{\mathbf{a}}
\newcommand{\bb}{\mathbf{b}}
\newcommand{\bk}{\mathbf{k}}
\newcommand{\bc}{\mathbf{c}}
\newcommand{\bx}{\mathbf{x}}
\newcommand{\bu}{\mathbf{u}}
\newcommand{\bv}{\mathbf{v}}
\newcommand{\bzero}{\mathbf{0}}

\newcommand{\bmu}{\boldsymbol{\mu}}
\newcommand{\bdel}{\boldsymbol{\delta}}
\newcommand{\bgam}{\boldsymbol{\gamma}}
\newcommand{\bth}{\boldsymbol{\theta}}
\newcommand{\bSig}{\boldsymbol{\Sigma}}
\newcommand{\bGam}{\boldsymbol{\Gamma}}
\newcommand{\bpsi}{\boldsymbol{\psi}}
\newcommand{\blam}{\boldsymbol{\lambda}}
\newcommand{\E}{\mathbb{E}}
\newcommand{\Var}{\operatorname{Var}}
\newcommand{\Cov}{\operatorname{Cov}}
\newcommand{\rank}{\operatorname{rank}}
\newcommand{\rowsp}{\operatorname{row}}
\newcommand{\suppp}{\operatorname{supp}}
\newcommand{\diag}{\operatorname{diag}}
\newcommand{\dto}{\overset{d}{\longrightarrow}}
\newcommand{\pto}{\overset{p}{\longrightarrow}}
\newcommand{\ind}{\perp\!\!\!\perp}
\newcommand{\cX}{\mathcal{X}}
\newcommand{\cW}{\mathcal{W}}
\newcommand{\cP}{\mathcal{P}}
\newcommand{\cI}{\mathcal{I}}
\newcommand{\cN}{\mathcal{N}}
\newcommand{\cM}{\mathcal{M}}

\title{Moments of Random Coefficients in Short Panels%
\thanks{Botosaru gratefully acknowledges financial support from the Canada Research Chairs Program.  During the preparation of this manuscript, the authors used Claude (Anthropic) to revise expository text and examples from the authors' theorems, proofs, and notes; check mathematical derivations; and check consistency of the LaTeX source. The substantive results and arguments originated with the authors.}}
\author{%
  Irene Botosaru\thanks{McMaster University, Department of Economics. Email: {\tt botosari@mcmaster.ca}},
  James L. Powell\thanks{University of Arizona, Department of Economics. Email: {\tt jlpowell@arizona.edu}}}
\date{August 31, 2026}

\begin{document}
\maketitle

\begin{abstract}
We study identification and estimation of moments of random coefficients in short linear panels, allowing the number of heterogeneous coefficients to exceed the number of equations observed for each unit.
Under moment homogeneity, different regressor histories impose restrictions on the same moment vector.
We give necessary and sufficient conditions for these restrictions to identify moments of a given order, stated in terms of the row spaces generated by the regressor support.
The results show that moments may be identified even when the coefficients cannot be recovered for any individual, and, for two full-row-rank histories, give the exact loss of independent restrictions caused by overlap of their row spaces.
When the support condition fails, we establish nonidentification in the maintained model and characterize the sharp identified set implied by these conditional moments.
At second order, the identified set is determined by positive-semidefinite covariance restrictions and is also sharp relative to the full joint distribution of outcomes and regressors.
Under the support condition, weighted minimum-distance estimators are root-$N$ asymptotically normal; under conditional nondegeneracy, oracle generalized-inverse weighting attains the Chamberlain (1987) efficiency bound for the maintained conditional-moment model.
\end{abstract}

\noindent\textbf{Keywords:} random coefficients; moment identification; panel data; partial identification; semidefinite programming; conditional moments.\\
\textbf{JEL codes:} C13, C14, C23.

\section{Introduction}\label{sec:intro}

Random coefficient models allow regression coefficients to vary across individuals.
We study identification and estimation of means, variances, covariances, and higher co-moments of random coefficients in short linear panels, allowing the number of heterogeneous coefficients to exceed the number of equations observed for each unit.
Within a unit, each equation reveals one linear combination of the coefficient vector, and the equations together reveal only coefficient combinations whose loading vectors lie in the row space of the unit's regressor matrix.
For a given regressor history, the conditional moments of the outcomes therefore restrict the moments of coefficient combinations associated with that row space.
Because regressor histories vary across units, the corresponding row spaces and restrictions can also vary.
Under moment homogeneity at a given order, the relevant conditional moments of the random coefficients are invariant across regressor histories.
The restrictions generated by different histories then apply to the same population moment vector and can be pooled.
Identification depends on whether the pooled restrictions determine all coefficient moments of that order.

Focusing on moments of a fixed order separates this finite-dimensional identification problem from recovery of the full coefficient distribution, which generally requires additional support or regularity conditions.
The panel problem is then to determine how many independent restrictions one regressor history contributes, when restrictions generated by different histories overlap, and whether variation across histories can identify population moments despite unit-level underdetermination.

Our first result gives a necessary and sufficient condition for point identification of the coefficient moments of a given order, stated in terms of the row spaces generated by the support of the regressors: the pooled restrictions have full column rank if and only if no nonzero homogeneous polynomial of that order vanishes on every row space in the support.
The condition is primitive in the sense that it refers only to the support of the regressor distribution, and it permits every unit-level system to be underdetermined: identification operates through variation of the row spaces across units, not through recovery of any unit's coefficient vector.
If the number of equations equals the number of random coefficients, one full-rank regressor history identifies moments of every order.
If every regressor matrix in the support has full row rank and the number of equations is one less than the number of random coefficients, moments of order $r$ are identified if and only if the support contains at least $r+1$ distinct row spaces.
In the cross-sectional model with a random intercept and a scalar random slope, this implies that $r+1$ distinct regressor values are necessary and sufficient for identification of the order-$r$ moments.
This recovers the sufficient support condition in Hermann and Holzmann (2025) and shows that it is also necessary.

The number of distinct regressor histories alone does not determine whether the rank condition holds.
Two full-row-rank histories may provide partly redundant restrictions if their row spaces overlap.
We derive an exact expression for the loss of independent restrictions generated by their intersection.
The identifying content of the panel therefore depends both on the dimension of each row space and on how the row spaces vary across regressor histories.
Simply counting equations or support points can overstate the information available about the coefficient moments.

The support conditions are nested across orders, so the condition at order $r$ is necessary and sufficient for joint identification of all moments through order $r$.
If the support condition fails, the conditional-moment equations leave a nontrivial affine class of candidate moment vectors: perturbing the true moments in any direction of the common null space yields a vector satisfying the same equations.
Rank failure alone does not establish nonidentification in the maintained model, because a perturbed vector must also be the moment vector of some coefficient distribution.
Theorem~\ref{thm:partialset} supplies this missing step by exhibiting coefficient distributions whose moment vectors differ in a null-space direction yet generate the same conditional moments of the outcomes through order $r$.
Failure of the support condition therefore implies nonidentification in the maintained model.
The same theorem characterizes the sharp identified set implied by the maintained conditional-moment restrictions: the affine class left unrestricted by those moments, intersected with the set of feasible moment vectors.
Linear functionals orthogonal to the unidentified directions remain identified.
At second order, feasibility is equivalent to positive semidefiniteness of the coefficient covariance matrix, so the sharp identified set is a spectrahedron, an affine slice of the positive semidefinite cone.
The resulting identified set is sharp relative to the joint distribution of outcomes and regressors, in the sense that features of the conditional distribution of the outcomes beyond its first two moments impose no additional restrictions on the mean and covariance of the random coefficients.
The second-order identified set is compact if and only if the first moments are identified, and sharp lower and upper bounds on linear functionals can be computed by semidefinite programming.
The binary-regressor variance bounds in Hermann and Holzmann (2025) arise as a closed-form special case.

Under the conditional nondegeneracy condition stated below, the same support condition is equivalent to nonsingularity of the information matrix.
When the support condition holds, estimation is a finite-dimensional conditional-moment problem: weighted minimum-distance estimators are consistent and root-$N$ asymptotically normal with heteroskedasticity-robust covariance estimators, and, under conditional nondegeneracy, weighting by the generalized inverse of the known conditional covariance matrix of the outcome moments attains the Chamberlain (1987) efficiency bound for the maintained conditional-moment model.
Feasible estimation of this optimal weight is not developed here.
Because identification at order $r$ implies identification at every lower order, the moments can be estimated jointly, which permits inference for raw moments and, through the polynomial transformation from raw to central moments, for means and central co-moments.

Moment homogeneity is weaker than independence of the coefficients and regressors, but it requires the coefficient moments under consideration to be invariant across regressor histories.
We consider two relaxations.
The first allows the conditional coefficient moments to vary linearly with known functions of the regressors.
Identification is then determined by the rank of the corresponding expanded system of conditional-moment restrictions; rich regressor support need not rule out observationally equivalent parameter values.
The second imposes moment homogeneity conditional on a control variable.
Estimation under these extensions is left for future work.

\subsection{Related literature}\label{sec:literature}

Within the econometric random-coefficient literature, Hermann and Holzmann (2025) provide the closest comparison for moment identification. 
They study a cross-sectional random coefficient regression with finitely supported regressors, derive conditions for identification of first and second moments, discuss higher-order mixed moments, and study partial identification with binary regressors. 
Their model corresponds to the cross-sectional case in which each regressor realization supplies one coefficient direction. 
In a panel or system of equations, a regressor history instead supplies an entire row space. 
This leads to additional restrictions involving the number and overlap of row spaces. 
One full-rank regressor matrix identifies moments of every fixed order when the number of equations equals the number of coefficients. 
When every regressor matrix in the support has full row rank and the number of equations is one less than the number of coefficients, $r+1$ distinct row spaces are necessary and sufficient for the order-$r$ support condition. 
Intersections among row spaces determine exactly how many independent order-$r$ restrictions are duplicated across regressor histories.
With a random intercept and a scalar regressor, our condition reduces to the $r+1$-support-point condition in Hermann and Holzmann (2025) and establishes its necessity.
Our second-order identified set also contains their binary-regressor variance bounds as a special case.
More generally, Theorem~\ref{thm:partialset} characterizes the sharp identified set when the support condition fails, and Corollary~\ref{cor:sharpdata} shows that at second order this set is sharp relative to the full observed distribution.

The paper is also related to the short-panel literature on correlated random coefficients.
Chamberlain (1992) studies average random coefficients and semiparametric efficiency while allowing unrestricted dependence between coefficients and regressors.
Graham and Powell (2012) identify average partial effects using stayers and the behavior of conditional means near singular regressor histories.
In their correlated random coefficient panel model, the semiparametric information for the average partial effect is singular even though the effect remains identified, so estimation is irregular.
By contrast, when $T<q$, each conditional information matrix in our model is singular; under conditional nondegeneracy and the support condition, aggregation across regressor histories yields a nonsingular  information matrix and regular root-$N$ estimation.
Arellano and Bonhomme (2012) identify variances and other distributional features by restricting the time-series dependence of idiosyncratic errors.
Lee (2026) studies dynamic random coefficient models with predetermined regressors and obtains partially identified means, variances, and distribution functions.
Our focus is instead on identification of fixed-order coefficient moments from cross-sectional variation in panel regressor histories under moment homogeneity.

The classical random coefficient regression literature already focused on low-order features of coefficient heterogeneity, including the mean and covariance matrix; see Hildreth and Houck (1968) and Swamy (1970). 
Later work studies nonparametric identification and estimation of the coefficient distribution; see Beran and Hall (1992), Beran, Feuerverger, and Hall (1996), Hoderlein, Klemel\"a, and Mammen (2010), Holzmann and Meister (2020), and Breunig (2021).  
These papers target nonparametric recovery of the coefficient distribution, an infinite-dimensional inverse problem, whereas our parameter is a fixed-order coefficient moment vector.

Gaillac and Gautier (2022) likewise study identification and estimation under limited regressor variation.  
They retain independence between the random coefficients and regressors and replace conventional large-support conditions with restrictions on the coefficient class, obtaining adaptive minimax density estimation.  
Gaillac and Gautier (2021) establish distributional identification under limited variation using quasi-analytic restrictions on the coefficient distribution.  
We impose no restrictions on the coefficient distribution; instead, we fix the moment order and ask whether the observed regressor histories identify those moments.  
Bounded or discrete support can therefore suffice.  
When the support condition for these moments fails, we characterize the resulting identified set rather than impose additional restrictions on the coefficient distribution to recover distributional point identification.

Masten (2018) targets the coefficient distribution in simultaneous-equation models using instruments, support conditions, and tail restrictions, whereas we target fixed-order coefficient moments under moment homogeneity and characterize identification from variation in regressor histories.

Related moment-recovery problems arise outside econometrics.
Bodmann, Ehler, and Gr\"af (2018) study reconstruction of moments of a common latent distribution from lower-dimensional linear measurements and characterize when a collection of measurement subspaces determines those moments; for a finite collection of measurements, their spanning condition is equivalent to the common-null-space condition used below.
Katsevich, Katsevich, and Singer (2015) estimate means and covariances of a latent signal from random linear projections.
The econometric problem differs in three respects.
First, the measurement matrices are realized regressor histories rather than designed projections, so identification rests on a support condition on the distribution of the regressors rather than on a chosen measurement design.
Second, the assumption of a common latent distribution across measurement matrices is replaced by moment homogeneity, which requires only the conditional coefficient moments of the maintained order to be invariant across regressor histories.
Third, when the common-null-space condition fails, establishing statistical nonidentification requires imposing moment feasibility on the affine alternatives; this yields the sharp identified-set results of Section~\ref{sec:partial}.
We also develop weighted minimum-distance estimation and an oracle efficiency benchmark for the resulting conditional-moment model.

Once the higher-order moment equations are formed, estimation is a finite-dimensional conditional-moment problem of the type studied by Chamberlain (1987). 
The identification analysis determines when the regressor support gives these equations full column rank.  
The efficiency result is therefore relative to the maintained conditional-moment model; we do not claim efficiency relative to all restrictions implied by the underlying structural random coefficient model.

The remainder of the paper is organized as follows.  
Section~\ref{sec:model} introduces the random coefficient system and develops the higher-order moment representation.  
Section~\ref{sec:identification} gives the necessary and sufficient support condition and its implications for short-panel regressor structures.  
Section~\ref{sec:estimation} develops estimation, asymptotic normality, robust inference, and oracle optimal weighting.  
Section~\ref{sec:stacked} treats joint estimation of moments through order \(r\), and Section~\ref{sec:central} develops the corresponding mean and central-moment parametrization.  
Section~\ref{sec:partial} establishes nonidentification in the maintained model when the support condition fails and characterizes the sharp identified sets.
Section~\ref{sec:discussion} discusses empirical diagnostics for regressor variation and extensions.  
Proofs and additional support results are collected in the appendices.

\section{Model and Moment Representation}\label{sec:model}

\subsection{Random coefficient system}\label{sec:setup}

We observe $\{(\bY_i,\bW_i):i=1,\ldots,N\}$ generated by
\begin{equation}\label{eq:model}
\underset{(T\times1)}{\bY_i}
=
\underset{(T\times q)}{\bW_i}
\underset{(q\times1)}{\bD_i},
\qquad T\le q,
\end{equation}
where $\bD_i$ is unobserved.
We refer to a realization $\bW_i=w$ as a \emph{regressor design}; it records the complete $T$-equation regressor history for unit $i$.
The row space $\rowsp(w)$ indexes the linear combinations of the random coefficients that can be formed from the equations observed under that regressor history.
The main case is $q>T$, in which $\bD_i$ cannot be recovered from $(\bY_i,\bW_i)$ even when $\bW_i$ has full row rank.
When $q=T$ and $\bW_i$ has full rank, $\bD_i$ is recovered exactly; this case provides a useful benchmark.
The identification analysis below concerns the distribution of $(\bY,\bW)$; the sampling scheme is imposed in Assumption~\ref{ass:sampling}.\footnote{Individual recovery of $\bD_i$ is not required for identification of its moments.
A given regressor history determines particular linear combinations of the coefficients, and different histories can determine different combinations.
Conditional-moment restrictions associated with individual histories may therefore be rank deficient while their counterparts, pooled across regressor histories, identify the coefficient moments.}

For a multi-index $\ba=(a_1,\ldots,a_T)'\in\mathbb{N}_0^T$, let
$|\ba|=\sum_{t=1}^T a_t$ and
$\bY^{\ba}=\prod_{t=1}^T Y_t^{a_t}$.  
For $\bb=(b_1,\ldots,b_q)'\in\mathbb{N}_0^q$, define analogously $|\bb|=\sum_{j=1}^q b_j$ and $\bD^{\bb}=\prod_{j=1}^qD_j^{b_j}$.  
Fix arbitrary orderings of the multi-indices of total degree $r$ and stack
\begin{align}
\bY^{[r]}&=\big(\bY^{\ba(1)},\ldots,\bY^{\ba(S_r)}\big)',
& S_r&=\binom{T+r-1}{r},\label{eq:Ylift}\\
\bD^{[r]}&=\big(\bD^{\bb(1)},\ldots,\bD^{\bb(Q_r)}\big)',
& Q_r&=\binom{q+r-1}{r},\label{eq:Dlift}\\
\bdel^{[r]}&=\E[\bD^{[r]}].\label{eq:deltar}
\end{align}
The orderings affect only the matrix representation, not the identifying content.

For $\bk=(k_1,\ldots,k_q)'\in\mathbb{N}_0^q$ with $|\bk|=a$, write $\binom{a}{\bk}=a!/(k_1!\cdots k_q!)$.  Let $\bW_t'$ denote the $t$th row of $\bW$ and $\bW_t^{\bk}=\prod_{j=1}^qW_{tj}^{k_j}$.

\subsection{Higher-order moment representation}\label{sec:representation}

\begin{lemma}[Higher-order moment representation]\label{lem:rep}
For every integer $r\ge1$ and every multi-index $\ba$ satisfying $|\ba|=r$,
\begin{equation}\label{eq:rep-entry}
\bY^{\ba}
=
\sum_{\bb:\,|\bb|=r}\bR_{r,\ba\bb}(\bW)\bD^{\bb},
\end{equation}
where
\begin{equation}\label{eq:Rentry}
\bR_{r,\ba\bb}(\bW)
=
\sum_{(\bk_1,\ldots,\bk_T)\in K_{\ba\bb}}
\prod_{t=1}^T
\binom{a_t}{\bk_t}\bW_t^{\bk_t},
\end{equation}
and
\begin{equation}\label{eq:Kab}
K_{\ba\bb}
=
\left\{(\bk_1,\ldots,\bk_T)\in(\mathbb{N}_0^q)^T:
|\bk_t|=a_t\ \forall t,\quad \sum_{t=1}^T\bk_t=\bb
\right\}.
\end{equation}
Consequently,
\begin{equation}\label{eq:lift}
\bY^{[r]}=\bR_r(\bW)\bD^{[r]},
\end{equation}
where $\bR_r(\bW)$ is an $S_r\times Q_r$ matrix whose entries are homogeneous polynomials of degree $r$ in the entries of $\bW$.
\end{lemma}

The matrix $\bR_r(\bW)$ is the order-$r$ matrix induced by the linear map $d\mapsto\bW d$.
Its rank and null space determine which order-$r$ coefficient moments are restricted by a given regressor history.
Equivalently, $\bR_r(\bW)$ represents the induced linear map on homogeneous polynomials of degree $r$, a representation used below to characterize the rank condition directly in terms of the regressor support.

\begin{lemma}[Rank of the induced order-$r$ matrix $\bR_r(\bW)$]\label{lem:rank}
For every $|\ba|=|\bb|=r$, the set $K_{\ba\bb}$ is nonempty.  If $\rank(\bW)=T$, then
\begin{equation}\label{eq:Rrank}
\rank\big(\bR_r(\bW)\big)=S_r.
\end{equation}
\end{lemma}

When $q>T$, Lemma~\ref{lem:rank} gives full \emph{row} rank but $S_r<Q_r$, so the order-$r$ system is underdetermined at every full-row-rank regressor history.
Identification therefore cannot be obtained by inverting the system at a single history.
It depends instead on whether the matrices associated with the histories in the regressor support have a common nonzero null direction.
The next section characterizes this condition in terms of the row spaces of the original regressor matrices.

\begin{example}[Computing an entry of $\bR_r$]\label{ex:entry}
Let $T=3$, $q=4$, $\ba=(2,1,0)'$, and $\bb=(0,1,2,0)'$, so $r=3$ and $\bY^{\ba}=Y_1^2Y_2$.  The set $K_{\ba\bb}$ contains exactly two elements,
\[
\big((0,0,2,0)',(0,1,0,0)',\bzero\big)
\qquad\text{and}\qquad
\big((0,1,1,0)',(0,0,1,0)',\bzero\big),
\]
with multinomial weights $\binom{2}{0,0,2,0}\binom{1}{0,1,0,0}=1$ and $\binom{2}{0,1,1,0}\binom{1}{0,0,1,0}=2$, so \eqref{eq:Rentry} gives
\begin{equation}\label{eq:entry-example}
\bR_{3,\ba\bb}(\bW)
=
W_{13}^2W_{22}
+
2\,W_{12}W_{13}W_{23},
\end{equation}
the coefficient on $\bD^{\bb}=D_2D_3^2$ in the expansion of
$Y_1^2Y_2$.
The first term results from selecting $D_3$ from both factors of $Y_1$ and
$D_2$ from $Y_2$.
The second results from selecting $D_2$ and $D_3$ from the two factors of
$Y_1$ and $D_3$ from $Y_2$, with multiplicity two.
\end{example}

\begin{remark}[Zero design entries and the identity design]\label{rem:zeroentries}
If $W_{tj}=0$, every term of \eqref{eq:Rentry} with $k_{tj}>0$ vanishes.  For $\bW=\bI_q$, the only surviving array concentrates each $\bk_t$ on coordinate $t$, forcing $\bb=\ba$, so $\bY^{\ba}=\bD^{\ba}$ and, under a common ordering of the multi-indices, $\bR_r(\bI_q)=\bI_{Q_r}$.  Rows that are coordinate vectors, as in Example~\ref{ex:runningA} below, likewise select single products $\bD^{\bb}$.
\end{remark}

\subsection{Two examples}\label{sec:examples}

\begin{example}[Scalar outcome]\label{ex:scalar}
Let $T=1$, $q=2$, and $\bW=(1,X)$.  Then $Y=D_0+XD_1$.  For $r=2$,
\begin{equation}\label{eq:scalar2}
Y^2
=
D_0^2+2X D_0D_1+X^2D_1^2,
\end{equation}
so
\[
\bR_2(X)=(1,2X,X^2),
\qquad
\bdel^{[2]}
=
\big(\E[D_0^2],\E[D_0D_1],\E[D_1^2]\big)'.
\]
With a scalar positive weight $A_2(X)$, the second-moment vector is identified
if and only if $1$, $X$, and $X^2$ are linearly independent in the weighted
$L^2$ space induced by the distribution of $X$. If $X$ has finite support, three distinct support points suffice.  More generally, for order $r$, $\bR_r(X)$ contains $1,X,\ldots,X^r$ up to multinomial coefficients, so $r+1$ distinct support points suffice to identify the $r+1$ order-$r$ mixed moments.  Hermann and Holzmann (2025) give this finite-support sufficient condition; Theorem~\ref{thm:hyper} below shows that, in this scalar case, the condition is also necessary.
\end{example}

\begin{example}[A running example: two underdetermined regressor histories]\label{ex:runningA}
For arbitrary $T\le q$, under an ordering consistent with
$\operatorname{vech}$, the representation at $r=2$ is
\begin{equation}\label{eq:vech}
\operatorname{vech}(\bY\bY')
=
\bR_2(\bW)\operatorname{vech}(\bD\bD').
\end{equation}
Thus each regressor history gives $T(T+1)/2$ equations relating
$q(q+1)/2$ distinct second-order coefficient products.
After taking conditional expectations and imposing second-order moment
homogeneity, these equations restrict the corresponding second
moments.
The following specialization recurs throughout the paper.

Let $q=3$ and $T=2$, and write $\bD=(D_1,D_2,D_3)'$.  Suppose the support of $\bW$ initially consists of
\[
\bW^{(1)}
=
\begin{pmatrix}
1&0&0\\
0&1&0
\end{pmatrix},
\qquad
\bW^{(2)}
=
\begin{pmatrix}
1&0&0\\
0&0&1
\end{pmatrix},
\]
each occurring with positive probability.  The corresponding systems are
\[
\bY
=
\begin{cases}
(D_1,D_2)', & \bW=\bW^{(1)},\\[2pt]
(D_1,D_3)', & \bW=\bW^{(2)}.
\end{cases}
\]
Each realized system contains two equations for three random coefficients, so $\bD$ cannot be recovered unit by unit.

Maintain first- and second-order moment homogeneity and write
\[
\bmu
=
\E[\bD]
=
(\mu_1,\mu_2,\mu_3)'
\]
and
\[
\E[\bD\bD']
=
\begin{pmatrix}
m_{11} & m_{12} & m_{13}\\
m_{12} & m_{22} & m_{23}\\
m_{13} & m_{23} & m_{33}
\end{pmatrix}.
\]
The conditional means identify
\[
\E[\bY\mid\bW^{(1)}]
=
\begin{pmatrix}\mu_1\\ \mu_2\end{pmatrix},
\qquad
\E[\bY\mid\bW^{(2)}]
=
\begin{pmatrix}\mu_1\\ \mu_3\end{pmatrix},
\]
so all three first moments are point identified.

At second order, use the ordering
\[
\bdel^{[2]}
=
(m_{11},m_{12},m_{13},m_{22},m_{23},m_{33})'.
\]
Then
\[
\bR_2(\bW^{(1)})
=
\begin{pmatrix}
1&0&0&0&0&0\\
0&1&0&0&0&0\\
0&0&0&1&0&0
\end{pmatrix},
\qquad
\bR_2(\bW^{(2)})
=
\begin{pmatrix}
1&0&0&0&0&0\\
0&0&1&0&0&0\\
0&0&0&0&0&1
\end{pmatrix}.
\]
Hence the first design identifies $m_{11}$, $m_{12}$, and $m_{22}$, and the second identifies $m_{11}$, $m_{13}$, and $m_{33}$.  Five of the six second moments are therefore identified.  The only moment not determined by these two systems is $m_{23}=\E[D_2D_3]$.

Thus neither regressor history identifies $\bD$ unit by unit, but together the two histories identify all first moments and five of the six second moments.
\end{example}

\section{Identification}\label{sec:identification}

Fix an integer $r\ge1$.
Under moment homogeneity, the higher-order moment representation gives
\[
\E[\bY^{[r]}\mid\bW=w]
=
\bR_r(w)\bdel^{[r]}.
\]
Identification of $\bdel^{[r]}$ therefore depends on whether the map
\[
\mathcal T_r(\delta)(w)
=
\bR_r(w)\delta,
\qquad w\in\cW,
\]
is injective over the regressor support.
We first express this requirement as a rank condition and then
characterize it directly in terms of the row spaces generated by the support
of the regressors.

\begin{assumption}[Structural equation]\label{ass:model}
Equation \eqref{eq:model} holds with $T\le q$.
\end{assumption}

\begin{assumption}[$r$th-order moment homogeneity]\label{ass:moment}
The vector $\bD^{[r]}$ is integrable and
\begin{equation}\label{eq:momenthom}
\E[\bD^{[r]}\mid\bW]
=
\E[\bD^{[r]}]
=
\bdel^{[r]}
\quad\text{a.s.}
\end{equation}
\end{assumption}

Assumption~\ref{ass:moment} is weaker than statistical independence
$\bD\ind\bW$.
It restricts only the order-$r$ conditional moments of $\bD$, not its full
conditional distribution.
Because $\bW$ contains the complete $T$-equation regressor history, the
restriction conditions on the entire history.
If instead
\[
m_r(w)=\E[\bD^{[r]}\mid\bW=w]
\]
were unrestricted in $w$, then
\[
\E[\bY^{[r]}\mid\bW=w]
=
\bR_r(w)m_r(w),
\]
so different regressor histories would generally restrict different
$Q_r$-vectors and their restrictions could not be pooled to identify a common
$\bdel^{[r]}$ when $q>T$.
Full row rank of $\bW$ is not required for the support characterization below;
it is imposed only when we count the maximum number of independent order-$r$
restrictions supplied by a single regressor history.

\begin{remark}[Finite-dimensional relaxation of moment homogeneity]
\label{rem:crc-relax}
Moment homogeneity can be relaxed without allowing $m_r(w)$ to vary
unrestrictedly.
Suppose, for a known $K$-vector of functions $z(\bW)$,
\begin{equation}\label{eq:crc-relax}
\E[\bD^{[r]}\mid\bW]=\bC_r z(\bW),
\qquad
\bC_r\in\mathbb R^{Q_r\times K}.
\end{equation}
Then
\begin{equation}\label{eq:crc-lift}
\E[\bY^{[r]}\mid\bW]
=
\bR_r(\bW)\bC_r z(\bW)
=
\bZ_r(\bW)\operatorname{vec}(\bC_r),
\qquad
\bZ_r(\bW)=z(\bW)'\otimes\bR_r(\bW).
\end{equation}
The conditional-moment model remains finite dimensional and linear in
$\operatorname{vec}(\bC_r)$.
Nonsingularity of
$\E[\bZ_r'\bA_r\bZ_r]$, for a symmetric and almost surely positive definite
weight $\bA_r$ of conformable dimension with
$\E\|\bZ_r'\bA_r\bZ_r\|<\infty$, is necessary and
sufficient for injectivity of the conditional-moment map in
$\operatorname{vec}(\bC_r)$.
The unconditional order-$r$ moment vector is then
\[
\bdel^{[r]}=\bC_r\E[z(\bW)],
\]
whenever the expectation exists, and the baseline model is the special case
$K=1$ and $z(\bW)\equiv1$.

Define the null space of the expanded conditional-moment matrix by
\[
\cN_r^{\bZ}
=
\{c:\bZ_r(w)c=\bzero\ \text{for all }w\in\suppp(\bW)\}.
\]
Its elements are perturbation directions in $\operatorname{vec}(\bC_r)$ left
unrestricted by the conditional-moment equations.
To conclude that such directions correspond to observationally equivalent
structures satisfying the maintained model, one would additionally need to verify feasibility of
the perturbed conditional-moment specifications, as is done for the baseline
model in Section~\ref{sec:partial}.

Rich support of $\bW$ need not eliminate $\cN_r^{\bZ}$.
When components of $z(\bW)$ also enter the regressor design, columns of
$\bZ_r(\bW)$ can be linearly dependent for every realization of $\bW$,
regardless of the richness of its support.
The simplest case is $T=1$, $r=1$, $\bW=(1,X)$, and
$z(\bW)=(1,X)'$, for which
\[
\bZ_1(x)
=
(1,x)\otimes(1,x)
=
(1,\;x,\;x,\;x^2).
\]
The four entries of $\operatorname{vec}(\bC_1)$ therefore multiply only three
linearly independent functions of $x$, and
$(0,1,-1,0)'\in\cN_1^{\bZ}$ for every distribution of $X$.
Only $C_{11}$, $C_{22}$, and the sum $C_{21}+C_{12}$ are identified,
however rich the support of $X$.

This specification allows the order-$r$ conditional moments of the
coefficients to vary systematically with $\bW$ while retaining a
finite-dimensional parameterization.
Identification is then governed by the rank of the expanded
conditional-moment system rather than by Theorem~\ref{thm:support}.
The expanded system can remain rank deficient even under arbitrarily rich
regressor support, so we do not pursue a primitive support characterization
for this extension.
\end{remark}

\begin{remark}[Control-variable relaxation]\label{rem:control-relax}
A different relaxation preserves the original support condition
conditionally.
Suppose an observed control variable $V$ satisfies
\begin{equation}\label{eq:control-relax}
\E[\bD^{[r]}\mid\bW,V]=m_r(V)
\end{equation}
for an unknown function $m_r$.
Then
\begin{equation}\label{eq:control-lift}
\E[\bY^{[r]}\mid\bW,V]
=
\bR_r(\bW)m_r(V).
\end{equation}
For each $v$, define the conditional effective regressor support
\begin{equation}\label{eq:control-cone}
\cX(v)
=
\bigcup_{w\in\suppp(\bW\mid V=v)}\rowsp(w).
\end{equation}
If, for almost every $v$, no nonzero degree-$r$ homogeneous polynomial
vanishes on $\cX(v)$, then the argument of Theorem~\ref{thm:support} applied
conditional on $V=v$ identifies $m_r(v)$ almost everywhere, and
\[
\bdel^{[r]}=\E[m_r(V)].
\]
This formulation allows the order-$r$ conditional moments of $\bD$ to vary
with $\bW$ through $V$.
For almost every $v$, identification is governed by the same support
condition applied to the conditional support of $\bW\mid V=v$.
Estimation of $m_r(v)$ is nonparametric when $V$ is continuous; if $V$ itself
is estimated, its first-stage error must also be accounted for.
The baseline model corresponds to a degenerate control.
\end{remark}

Call a measurable $S_r\times S_r$ matrix $\bA_r(\bW)$ \emph{admissible} if
it is symmetric and positive semidefinite almost surely, positive definite on
the column space of $\bR_r(\bW)$ almost surely, and
\[
\E\|\bR_r'\bA_r\bR_r\|<\infty
\qquad\text{and}\qquad
\E\|\bR_r'\bA_r\bY^{[r]}\|<\infty;
\]
under Assumption~\ref{ass:sampling} below, these two integrability conditions
are implied by \eqref{eq:int1}--\eqref{eq:int2}.
An almost surely positive-definite weight satisfying these integrability
conditions is admissible.
The weaker positive-definiteness requirement on
$\operatorname{col}(\bR_r(\bW))$ is needed for the oracle weight of
Section~\ref{sec:efficiency}, which is positive definite on this column space
but singular when $\bR_r(\bW)$ does not have full row rank.

For an admissible weight, define
\begin{equation}\label{eq:Mdef}
\bM_{r\bA}
=
\E\big[\bR_r(\bW)'\bA_r(\bW)\bR_r(\bW)\big].
\end{equation}
The rank condition is nonsingularity of $\bM_{r\bA}$.
Because $\bA_r(\bW)$ is positive definite on
$\operatorname{col}(\bR_r(\bW))$, Theorem~\ref{thm:support} shows that this
condition holds if and only if there is no nonzero vector $c$ such that
\[
\bR_r(w)c=\bzero
\qquad
\text{for every }w\in\suppp(\bW).
\]
The theorem characterizes this common-null-space condition directly in terms
of the row spaces generated by the regressor support.

Theorem~\ref{thm:partialset} (whose proof applies verbatim at a single order under Assumption~\ref{ass:moment}) establishes the converse identification statement:
when this common null space is nontrivial, the maintained model contains
structures with the same conditional moments and different
coefficient moments.
The rank condition is therefore necessary and sufficient for
point identification of $\bdel^{[r]}$ under the maintained
conditional-moment restrictions.
Section~\ref{sec:partial} characterizes the identified set when the condition
fails and imposes moment feasibility on the affine set left unrestricted by
the conditional moments.

\subsection{From the rank condition to a support condition}
\label{sec:support}

\begin{definition}[Effective regressor support]\label{def:cone}
Let $\cW=\suppp(\bW)\subseteq\mathbb R^{T\times q}$ and define
\begin{equation}\label{eq:cone}
\cX
=
\bigcup_{w\in\cW}\rowsp(w),
\qquad
\rowsp(w)=\{w'\blam:\blam\in\mathbb R^T\}\subseteq\mathbb R^q.
\end{equation}
The set $\cX$ contains every coefficient-loading vector that can be formed as
a linear combination of the rows of a regressor history in the support.
\end{definition}

To express the common-null-space condition directly in terms of the regressor
support, associate each $\bc\in\mathbb R^{Q_r}$ with the homogeneous
degree-$r$ polynomial
\begin{equation}\label{eq:Pc}
P_{\bc}(\bx)
=
\sum_{|\bb|=r}\binom{r}{\bb}c_{\bb}\bx^{\bb},
\qquad \bx\in\mathbb R^q,
\end{equation}
where $\binom r\bb=r!/\prod_j b_j!$.
The vector $\bc$ represents a direction in the order-$r$ coefficient-moment
parameter space.
Lemma~\ref{lem:polar} shows that
$\bR_r(w)\bc=\bzero$ if and only if $P_{\bc}$ vanishes on
$\rowsp(w)$.
The map $\bc\mapsto P_{\bc}$ is a linear bijection from
$\mathbb R^{Q_r}$ onto the space $\cP_r$ of homogeneous degree-$r$
polynomials on $\mathbb R^q$.
It therefore converts the common-null-space condition for $\bR_r(w)$ into a
condition on the row spaces generated by the regressor support.

\begin{lemma}[Polarization identity]\label{lem:polar}
For every $w\in\mathbb R^{T\times q}$, $\bc\in\mathbb R^{Q_r}$, and $\blam\in\mathbb R^T$,
\begin{equation}\label{eq:polar}
P_{\bc}(w'\blam)
=
\sum_{|\ba|=r}\binom r\ba
\big(\bR_r(w)\bc\big)_{\ba}\blam^{\ba}.
\end{equation}
Consequently, $P_{\bc}$ vanishes identically on $\rowsp(w)$ if and only if $\bR_r(w)\bc=\bzero$.
\end{lemma}

\begin{theorem}[Primitive support characterization]\label{thm:support}
Assume the class of admissible weights is nonempty.  The following statements are equivalent:
\begin{enumerate}[label=(\roman*)]
\item $\bM_{r\bA}$ is nonsingular for every admissible $\bA_r$;
\item $\bM_{r\bA}$ is nonsingular for at least one admissible $\bA_r$;
\item the only $\bc\in\mathbb R^{Q_r}$ satisfying $\bR_r(w)\bc=\bzero$ for every $w\in\cW$ is $\bc=\bzero$;
\item no nonzero homogeneous polynomial of degree $r$ on $\mathbb R^q$ vanishes identically on $\cX$.
\end{enumerate}
Moreover, the unidentified moment space
\begin{equation}\label{eq:Nr}
\cN_r
=
\{\bc:\bR_r(w)\bc=\bzero\ \text{for every }w\in\cW\}
\end{equation}
is linearly isomorphic, through $\bc\mapsto P_{\bc}$, to the space $\cI_r(\cX)$ of degree-$r$ homogeneous polynomials that vanish on $\cX$.
\end{theorem}

Statements (i)--(iii) express the rank condition as the absence of
a nonzero direction in the null space of $\bR_r(w)$ for every regressor
history in the support.
Statement (iv) gives the corresponding characterization in terms of the
original regressor support: such a common null direction exists if and only
if the associated homogeneous degree-$r$ polynomial vanishes on every row
space contained in $\cX$.
Because these conditions are equivalent for every admissible weight, the
choice of weight does not affect identification under the maintained
conditional-moment restrictions, although it can affect the asymptotic
variance of the resulting estimator.

\begin{remark}[Full row rank is not required]\label{rem:noR1}
Theorem~\ref{thm:support} does not require
$\Pr\{\rank(\bW)=T\}=1$.
A rank-deficient regressor history contributes its lower-dimensional row
space to $\cX$.
More generally, if $\rank(w)=k$, then
\begin{equation}\label{eq:general-rank}
\rank\bR_r(w)=\binom{k+r-1}{r}.
\end{equation}
To see this, the rows of $\bR_r(w)$ span the degree-$r$ products of the
linear forms $d\mapsto w_t'd$.
If $\rank(w)=k$, these linear forms span a $k$-dimensional space.
Their degree-$r$ products therefore span the space of homogeneous
degree-$r$ polynomials in $k$ variables, whose dimension is
$\binom{k+r-1}{r}$.
The same rank calculation appears in the proof of
Theorem~\ref{thm:efficient} through Lemma~\ref{lem:polar}.
Full row rank is therefore needed only when counting the maximum number of
independent order-$r$ restrictions supplied by a single regressor history,
not for the support characterization itself.
The support may consequently include rank-deficient histories, including
histories with no within-unit regressor variation.
\end{remark}

\begin{theorem}[Identification of order-$r$ moments]\label{thm:ident}
Under Assumptions~\ref{ass:model} and \ref{ass:moment}, suppose the equivalent conditions of Theorem~\ref{thm:support} hold.  Then, for every admissible weight,
\begin{equation}\label{eq:ident}
\bdel^{[r]}
=
\Big(\E[\bR_r'\bA_r\bR_r]\Big)^{-1}
\E[\bR_r'\bA_r\bY^{[r]}].
\end{equation}
Thus $\bdel^{[r]}$ is point identified even when $q>T$ and $\bD$ is not recoverable for any observation.
\end{theorem}

\begin{remark}[Moment versus distributional identification]\label{rem:finiteorder}
Theorem~\ref{thm:ident} identifies a fixed finite-dimensional moment vector.
It does not identify the distribution of $\bD$.
Identification of a distribution from its full moment sequence additionally
requires moment determinacy, while nonparametric density recovery generally
requires further regularity and entails an inverse problem.
These requirements are separate from the fixed-order moment-identification
problem studied here.
\end{remark}

\subsection{Implications of panel structure for identification}
\label{sec:panelgain}

Suppose a regressor history has full row rank $T$.
Lemma~\ref{lem:rank} implies that it supplies
\begin{equation}\label{eq:Sr-info}
S_r=\binom{T+r-1}{r}
\end{equation}
linearly independent order-$r$ restrictions.
When $T=1$, a full-row-rank history supplies only one such restriction.
If every support matrix has full row rank $T$ and the support contains $m$
distinct row spaces, a necessary counting condition for identification is
therefore
\begin{equation}\label{eq:count-bound}
mS_r\ge Q_r,
\qquad
Q_r=\binom{q+r-1}{r}.
\end{equation}
The condition is not sufficient because restrictions associated with
different row spaces can be linearly dependent.

\begin{corollary}[Identification from a single full-rank regressor history]
\label{cor:withinfull}
If $T=q$ and the support contains one matrix $w_0$ with $\rank(w_0)=q$, then $\rowsp(w_0)=\mathbb R^q$, so the support condition of Theorem~\ref{thm:support} holds at every order $r$.  No variation across regressor histories is required.
\end{corollary}

Thus, when $T=q$, one full-rank regressor history suffices; variation across regressor histories is not required for identification. 
When $T=1$, by contrast, each regressor history generates a one-dimensional
row space.
Theorem~\ref{thm:support} then reduces to the cross-sectional condition that
the regressor support not lie in the zero set of a nonzero homogeneous
degree-$r$ polynomial.
With an intercept, dehomogenization gives the corresponding condition that
the support of the nonconstant regressors not lie in the zero set of a
polynomial of degree at most $r$.\footnote{The aggregation of rank-one
contributions into a nonsingular matrix is classical in linear
regression and optimal experimental design (Elfving, 1952; Karlin and
Studden, 1966).  In polynomial regression, the $T=1$ problem is closely
related to the requirement that the design support span the relevant
polynomial regression functions.  Here the focus is the system case $T>1$,
where each regressor history generates a higher-dimensional row space and
identification depends on whether the resulting row spaces have a common
order-$r$ null direction.}

\begin{theorem}[Hyperplane row spaces]\label{thm:hyper}
Suppose every support matrix has row space equal to a hyperplane of $\mathbb R^q$; in particular, this is the full-row-rank case $T=q-1$.  Then the order-$r$ support condition of
Theorem~\ref{thm:support} holds if and only if the effective regressor
support contains at least $r+1$ distinct hyperplanes.
\end{theorem}

If a homogeneous polynomial vanishes identically on a hyperplane, the linear
form defining that hyperplane divides the polynomial.
A nonzero polynomial of degree $r$ therefore cannot vanish identically on
more than $r$ distinct hyperplanes.
This gives the $r+1$ threshold in Theorem~\ref{thm:hyper}.

\begin{remark}[Hermann--Holzmann as a special case]\label{rem:HHnest}
For $q=2$ and $T=1=q-1$, the hyperplanes are the lines through the
regressor vectors $(1,x_j)$.  Theorem~\ref{thm:hyper} therefore implies
that the order-$r$ moments are identified under the maintained conditional-moment model if and only if the scalar regressor has at least $r+1$ distinct support points.  This recovers the $r+1$-point sufficient condition of Proposition~2.2 in
Hermann and Holzmann (2025) and establishes its necessity in the scalar
case: with $r$ or fewer support points the rank condition fails.
At second order, the necessity of the rank condition among full-rank
covariance matrices is already in Theorem~2.3 of Hermann and Holzmann
(2025), whose proof perturbs a positive definite covariance along a null
direction of the design; Theorem~\ref{thm:partialset} below extends this
argument to every order and to the full stacked moment vector.\footnote{
Several of their results are instances of Theorem~\ref{thm:support}(iv)
with $T=1$.  Their Example~2.1, in which some regressor $W_j$ has only two
support points $a$ and $b$, is the statement that the homogeneous quadratic
$(x_j-a\,x_0)(x_j-b\,x_0)$ vanishes on every $(1,w')'$ with
$w_j\in\{a,b\}$, so $\cN_2\neq\{0\}$ whatever the support of the remaining
regressors.  Their Theorem~2.4, by which a Cartesian product of three points
in each coordinate identifies means and covariances, is the $r=2$ case of
(iv): a polynomial of degree at most two in each coordinate that vanishes
on such a grid is zero.  Their Theorem~2.6 gives a sufficient support
condition for higher-order moments, whereas condition (iv) is necessary
and sufficient.}
\end{remark}

For two full-row-rank regressor histories, the number of independent
restrictions they supply jointly also depends on the dimension of the
intersection of their row spaces.

\begin{proposition}[Overlap between two full-row-rank histories]\label{prop:twopoint}
Let $w^{(1)}$ and $w^{(2)}$ have full row rank $T$, with row spaces $L_1$ and $L_2$, and let $t=\dim(L_1\cap L_2)$.  Then
\begin{equation}\label{eq:twopoint}
\rank
\begin{pmatrix}
\bR_r(w^{(1)})\\
\bR_r(w^{(2)})
\end{pmatrix}
=
2\binom{T+r-1}{r}
-
\binom{t+r-1}{r},
\end{equation}
where $\binom{r-1}{r}=0$ when $t=0$.  The final term is the exact redundancy generated by the intersection of the two row spaces.
\end{proposition}

The intersection has dimension $t$ and therefore generates
$\binom{t+r-1}{r}$ order-$r$ restrictions common to the two histories.
These restrictions are counted twice in $2S_r$, producing the correction
term in \eqref{eq:twopoint}.

\begin{examplecont}{Example~\ref{ex:runningA} (continued: overlap and the unidentified direction)}
Since $q=3$, $T=2$, and $r=2$,
\[
Q_2=\binom{4}{2}=6,
\qquad
S_2=\binom{3}{2}=3.
\]
The counting condition \eqref{eq:count-bound} therefore allows for the
possibility that two regressor histories identify the six second moments.
For the two histories in Example~\ref{ex:runningA}, however, the row spaces
are
\[
L_1
=
\rowsp(\bW^{(1)})
=
\{\bx\in\mathbb R^3:x_3=0\},
\qquad
L_2
=
\rowsp(\bW^{(2)})
=
\{\bx\in\mathbb R^3:x_2=0\}.
\]
Their intersection is the line spanned by $(1,0,0)'$, so $t=1$ and Proposition~\ref{prop:twopoint} gives
\[
\rank
\begin{pmatrix}
\bR_2(\bW^{(1)})\\
\bR_2(\bW^{(2)})
\end{pmatrix}
=
2S_2-\binom{1+2-1}{2}
=
6-1
=
5.
\]
This rank deficiency is not specific to the particular orientation of the
two row spaces.
Any two planes in $\mathbb R^3$ intersect in at least a line, so $t\ge1$
and two full-row-rank histories cannot satisfy the second-order rank
condition in this configuration.

The same rank failure follows from Theorem~\ref{thm:support}.
The nonzero quadratic
\[
P(\bx)=x_2x_3
\]
vanishes on both $L_1$ and $L_2$.
Because the $m_{23}$ column of both matrices
$\bR_2(\bW^{(1)})$ and $\bR_2(\bW^{(2)})$ in
Example~\ref{ex:runningA} is identically zero, the unidentified moment
space is exactly
\[
\cN_2
=
\operatorname{span}\{(0,0,0,0,1,0)'\},
\]
and under the convention \eqref{eq:Pc} the vanishing polynomial associated with this coefficient direction is $P_{\bc}(\bx)=2x_2x_3$, proportional to $P$.

In this example, neither regressor history reveals $D_2$ and $D_3$ jointly.
The only second moment left unrestricted by the conditional-moment equations
is therefore
\[
\E[D_2D_3].
\]
The counting condition fails to detect this restriction because the
one-dimensional intersection of the two row spaces generates one duplicated
order-$2$ restriction.
\end{examplecont}

The role of the panel dimension can therefore be stated directly.
When $T=1$, each regressor history supplies one order-$r$ restriction, and
identification depends entirely on variation across the resulting
one-dimensional row spaces.
When $1<T<q$, each history supplies multiple restrictions, but identification
also depends on how its row space differs from those generated by other
histories.
When $T=q$, a single full-rank history supplies enough restrictions to
identify the moment vector at every fixed order.
In all three cases, the necessary and sufficient condition is the support
condition in Theorem~\ref{thm:support}.

\subsection{Panel structure with an intercept}\label{sec:intercept}

Suppose the rows take the form $\bW_t'=(1,\bX_t')$, with
$\bX_t\in\mathbb R^{q-1}$.
For a realization $\omega$, let
\[
U(\omega)
=
\operatorname{span}
\{\bX_t(\omega)-\bX_1(\omega):t=2,\ldots,T\}
\]
and let $\bar\bx$ be any point of
$\operatorname{aff}\{\bX_1(\omega),\ldots,\bX_T(\omega)\}$.
Then
\begin{equation}\label{eq:row-intercept}
\rowsp(\bW(\omega))
=
\{(s,s\bar\bx+\bu):s\in\mathbb R,\ \bu\in U(\omega)\}.
\end{equation}
For $s\neq0$, normalizing the first coordinate to one gives
$(1,\bar\bx+\bu/s)$, so the nonconstant regressor component ranges over
the affine hull of the within-unit regressor history.
For $s=0$, the row-space elements are $(0,\bu)$ with
$\bu\in U(\omega)$ and are generated by within-unit regressor differences.

\begin{proposition}[Intercept decomposition]\label{prop:intercept}
Let $P$ be a homogeneous polynomial of degree $r$ on $\mathbb R^q$ and let
$p(\bx)=P(1,\bx)$ be its dehomogenization, with highest-degree homogeneous
component $p_r$.
Then $P$ vanishes on the effective regressor support if and only if, for
every support realization $\omega$,
\begin{equation}\label{eq:intercept-cond}
p\equiv0
\quad\text{on}\quad
\operatorname{aff}\{\bX_1(\omega),\ldots,\bX_T(\omega)\},
\end{equation}
in which case the $s=0$ directions also satisfy
$p_r\equiv0$ on $U(\omega)$.
The condition on $p_r$ is therefore implied by the affine-hull condition
and is not an additional restriction.
Thus the support condition can be stated directly in terms of the affine
hulls generated by the within-unit regressor histories.
\end{proposition}

\begin{corollary}[Intercept plus scalar regressor]\label{cor:scalar}
Let $q=T=2$, so $\bW_t'=(1,X_t)$.
Under full row rank, $X_1\neq X_2$ and
Corollary~\ref{cor:withinfull} gives identification of moments of every
order from a single full-rank regressor history.
More generally, because Theorem~\ref{thm:support} does not require full row
rank almost surely, moments of every order are identified whenever
\begin{equation}\label{eq:positive-fullrank}
\Pr\{X_1\neq X_2\}>0.
\end{equation}
In particular, a binary scalar regressor with positive probability of
within-unit movement identifies moments of every fixed order of
$(D_0,D_1)$, whereas in the $T=1$ cross section a binary regressor does not
identify the full second-moment vector.
\end{corollary}

\begin{corollary}[No within variation: cross-sectional benchmark]
\label{cor:nowithin}
Suppose, explicitly relaxing full row rank, that
$\bX_1=\cdots=\bX_T$ almost surely.
Then $\rank(\bW)=1$ and each row space is the one-dimensional subspace
spanned by $(1,\bX_1')'$.
The effective regressor support therefore coincides with the $T=1$
effective support, and the identification condition reduces exactly to the
corresponding cross-sectional support condition.
Repeated observations with no within-unit regressor variation therefore add
no identifying restrictions beyond those available in the corresponding
cross section.
\end{corollary}

\subsection{Identification through several moment orders}
\label{sec:toporder}

The support conditions are nested across moment orders.

\begin{proposition}[Top-order dominance]\label{prop:toporder}
If the support condition of Theorem~\ref{thm:support} holds at order $r$, it
holds at every order $s\le r$.
Consequently, for the block-diagonal stacked matrix
$\bR_{\le r}=\diag(\bR_1,\ldots,\bR_r)$, the stacked Gram matrix
is nonsingular for every admissible stacked weight (admissibility being
defined as in Section~\ref{sec:identification} with $\bR_{\le r}$ in place
of $\bR_r$) if and only if the support condition holds at the single highest
order $r$.
\end{proposition}

Proposition~\ref{prop:toporder} implies that joint identification of moments
through order $r$ requires no support condition beyond the condition at
order $r$.
In particular, identification at order $r$ implies identification of the
lower-order moments needed to form means and central moments.
To see the nesting directly, suppose a nonzero homogeneous polynomial of
degree $s<r$ vanishes on the effective regressor support.
Multiplying it by a suitable power of any nonzero linear form produces a
nonzero homogeneous polynomial of degree $r$ that also vanishes on the
effective regressor support.
Failure of the support condition at any lower order therefore implies
failure at every higher order.

\section{Estimation and inference}\label{sec:estimation}

Let $\bA_{r,i}=\bA_r(\bW_i)$ and
$\bR_{r,i}=\bR_r(\bW_i)$.
Throughout this section, the weight function is treated as known.
This includes fixed choices such as $\bA_r=\bI_{S_r}$; the oracle optimal
weight is considered in Section~\ref{sec:efficiency}.

\begin{assumption}[Sampling and integrability]\label{ass:sampling}
The observations $(\bY_i,\bW_i)$ are i.i.d.  In addition,
\begin{align}
\E\|\bR_r'\bA_r\bR_r\|&<\infty,\label{eq:int1}\\
\E\left\|\bR_r'\bA_r\big(\bY^{[r]}-\bR_r\bdel^{[r]}\big)\right\|^2&<\infty.\label{eq:int2}
\end{align}
\end{assumption}

Condition~\eqref{eq:int1} gives the law of large numbers for the sample
Gram matrix.
Together with the finite first moment of
$\bR_r'\bA_r(\bY^{[r]}-\bR_r\bdel^{[r]})$, which follows from
\eqref{eq:int2}, it also gives the integrability required for consistency.
Condition~\eqref{eq:int2} imposes a finite second moment on the summand
entering the central limit theorem.
The sandwich covariance estimator requires an additional second-moment
condition, imposed only in Section~\ref{sec:robust}.
For bounded weights, Appendix~\ref{app:moments} gives primitive sufficient
conditions separately for consistency and for asymptotic normality and
sandwich covariance estimation: \eqref{eq:primitive-cons} implies the
consistency requirements, while \eqref{eq:primitive} implies those used for
the central limit theorem and the sandwich estimator.
Stating the required moments directly allows the dependence between $\bD$
and $\bW$ permitted by Assumption~\ref{ass:moment}.

On the event that $\widehat{\bM}_{r\bA}$ is nonsingular, define the weighted
minimum-distance estimator
\begin{equation}\label{eq:estimator}
\widehat{\bdel}^{[r]}
=
\widehat{\bM}_{r\bA}^{-1}
\frac1N\sum_{i=1}^N\bR_{r,i}'\bA_{r,i}\bY_i^{[r]},
\qquad
\widehat{\bM}_{r\bA}
=
\frac1N\sum_{i=1}^N\bR_{r,i}'\bA_{r,i}\bR_{r,i}.
\end{equation}
Under the conditions of Theorem~\ref{thm:asynorm}, this event has
probability approaching one.

At first and second order, closely related estimators appear in the
statistics literature on random linear projections.
Katsevich, Katsevich, and Singer (2015) estimate the mean and covariance of
a latent random vector from noisy random projections.
Their mean estimator coincides with \eqref{eq:estimator} at $r=1$ under
identity weighting.
After subtraction of the known noise variance, their covariance estimator
is the corresponding second-order minimum-distance estimator under the
Frobenius metric.
Their consistency analysis requires invertibility of the relevant
projection operators.
In the notation used here, the corresponding issue is nonsingularity of the
Gram matrix, for which Theorem~\ref{thm:support} gives the
support characterization.

\begin{remark}[Nonrandom designs]\label{rem:fixeddesign}
If the regressor matrices $w_1,\ldots,w_N$ are fixed rather than sampled,
the analysis conditions on the design sequence.
The estimator \eqref{eq:estimator} is unchanged.
If the coefficient vectors are independent across units and
$\E[\bD_i^{[r]}]=\bdel^{[r]}$, the moment contributions
\[
\bpsi_{r,i}
=
\bR_r(w_i)'\bA_{r,i}\bR_r(w_i)
\big(\bD_i^{[r]}-\bdel^{[r]}\big)
\]
are independent but need not be identically distributed.
Consistency and asymptotic normality are then triangular-array results.
They follow under the corresponding laws of large numbers and central limit
theorems, including convergence of the averaged Gram and moment-contribution
covariance matrices and a Lindeberg condition on $\{\bpsi_{r,i}\}$, together
with the design-sequence rank condition
\begin{equation}\label{eq:fixeddesign-rank}
\liminf_{N\to\infty}
\lambda_{\min}\!\left(
\frac1N\sum_{i=1}^N\bR_r(w_i)'\bA_{r,i}\bR_r(w_i)
\right)>0.
\end{equation}

Condition~\eqref{eq:fixeddesign-rank} replaces the support
condition of Theorem~\ref{thm:support}.
It fails if and only if there exist a subsequence $N_j$ and unit vectors
$\bc_{N_j}$ such that
\[
\bc_{N_j}'
\left(
\frac{1}{N_j}
\sum_{i=1}^{N_j}
\bR_r(w_i)'\bA_{r,i}\bR_r(w_i)
\right)
\bc_{N_j}
\longrightarrow0.
\]
Through the correspondence $\bc\mapsto P_{\bc}$ in
Section~\ref{sec:support}, these vectors correspond to degree-$r$
polynomials whose weighted evaluations over the design sequence become
arbitrarily small.
If a finite set of regressor histories is repeated with stable positive
frequencies, \eqref{eq:fixeddesign-rank} reduces to the corresponding
rank condition.

Choosing the $w_i$ to maximize the minimum eigenvalue in
\eqref{eq:fixeddesign-rank} is an $E$-optimal design problem (Ehrenfeld,
1955; Pukelsheim, 1993).
After fixing the normalization of the regressor histories and the weighting
rule, their row spaces summarize the rank component of the design problem;
Appendix~\ref{app:geometry} relates it to explicit fixed-design
constructions in the moment-recovery literature.
We do not develop the fixed-design asymptotic theory further.
\end{remark}

\begin{theorem}[Consistency and asymptotic normality]\label{thm:asynorm}
Let $\bA_r$ be a known admissible weight.  Under Assumptions~\ref{ass:model}, \ref{ass:moment}, and \ref{ass:sampling}, and the equivalent support conditions of Theorem~\ref{thm:support},
\begin{equation}\label{eq:consistency}
\widehat{\bdel}^{[r]}\pto\bdel^{[r]},
\end{equation}
and
\begin{equation}\label{eq:clt}
\sqrt{N}\big(\widehat{\bdel}^{[r]}-\bdel^{[r]}\big)
\dto
N\big(\bzero,\bM_{r\bA}^{-1}\bV_{r\bA}\bM_{r\bA}^{-1}\big),
\end{equation}
where
\begin{align}
\bpsi_{r}
&=
\bR_r'\bA_r\big(\bY^{[r]}-\bR_r\bdel^{[r]}\big),\label{eq:psi}\\
\bV_{r\bA}
&=
\E[\bpsi_r\bpsi_r'].
\label{eq:Vpsi}
\end{align}
Equivalently, letting
\begin{equation}\label{eq:Gamma}
\bGam_r(\bW)
=
\Var\big(\bY^{[r]}\mid\bW\big),
\end{equation}
which is finite almost surely under \eqref{eq:int2}, we have
\begin{equation}\label{eq:V}
\bV_{r\bA}
=
\E\big[\bR_r'\bA_r\bGam_r\bA_r\bR_r\big].
\end{equation}
Whenever $\bSig_r(\bW)=\Var(\bD^{[r]}\mid\bW)$ exists as a finite
matrix, the structural equation implies
$\bGam_r(\bW)=\bR_r\bSig_r(\bW)\bR_r'$; \eqref{eq:int2} controls
only the conditional second moments of the observed combinations
$\bR_r(\bW)\bD^{[r]}$ and does not require finite conditional second
moments of $\bD^{[r]}$ in the null space of $\bR_r(\bW)$.
Consistency \eqref{eq:consistency} does not require the full strength of \eqref{eq:int2}: it holds with \eqref{eq:int2} weakened to $\E\|\bpsi_r\|<\infty$.
\end{theorem}

\subsection{Robust covariance estimation}\label{sec:robust}

To establish consistency of the covariance estimator based on the estimated
moment contributions, impose the additional moment condition
\begin{equation}\label{eq:int3}
\E\|\bR_r'\bA_r\bR_r\|^2<\infty.
\end{equation}
Condition~\eqref{eq:int3} is not required for consistency or asymptotic
normality in Theorem~\ref{thm:asynorm}.
It controls the terms involving
$\widehat{\bdel}^{[r]}-\bdel^{[r]}$ that arise when the estimated residuals
are substituted into the empirical covariance matrix.

Define the estimated residual and moment contribution by
\begin{align}
\widehat e_{r,i}
&=
\bY_i^{[r]}-\bR_{r,i}\widehat{\bdel}^{[r]},\label{eq:resid}\\
\widehat\bpsi_{r,i}
&=
\bR_{r,i}'\bA_{r,i}\widehat e_{r,i},\label{eq:psihat}
\end{align}
and let
\begin{equation}\label{eq:Vhat}
\widehat\bV_{r\bA}
=
\frac1N\sum_{i=1}^N
\widehat\bpsi_{r,i}\widehat\bpsi_{r,i}'.
\end{equation}

\begin{theorem}[Robust covariance estimator]\label{thm:robust}
Under the assumptions of Theorem~\ref{thm:asynorm} and the additional moment condition \eqref{eq:int3},
\begin{equation}\label{eq:Vhatcons}
\widehat\bV_{r\bA}\pto\bV_{r\bA},
\end{equation}
and therefore
\begin{equation}\label{eq:sandwichhat}
\widehat\bM_{r\bA}^{-1}
\widehat\bV_{r\bA}
\widehat\bM_{r\bA}^{-1}
\pto
\bM_{r\bA}^{-1}\bV_{r\bA}\bM_{r\bA}^{-1}.
\end{equation}
\end{theorem}

\subsection{Oracle efficiency benchmark}\label{sec:efficiency}

We use the conditional covariance matrix
\[
\bGam_r(\bW)
=
\Var\big(\bY^{[r]}\mid\bW\big)
\]
to define an oracle efficiency benchmark for estimation of
$\bdel^{[r]}$ from the maintained conditional-moment restriction.
If $\bGam_r(\bW)$ were known, the corresponding optimal weight is its
Moore--Penrose generalized inverse $\bGam_r(\bW)^+$; when
$\bGam_r(\bW)$ is nonsingular, this reduces to the ordinary inverse.
This weight is infeasible because $\bGam_r(\bW)$ is unknown.
A feasible version would require estimation of the conditional covariance
matrix and additional regularity conditions, which we do not develop here.
The fixed-weight estimator and robust covariance estimator of the preceding
sections do not require this first step.

\begin{assumption}[Conditional nondegeneracy]\label{ass:nondeg}
$\bSig_r(\bW)=\Var(\bD^{[r]}\mid\bW)$ is positive definite almost surely, and
\begin{equation}\label{eq:Bdef}
\bB_r
=
\E\big[\bR_r'\bGam_r^{+}\bR_r\big]
\end{equation}
is finite, where $\bGam_r^{+}=\bGam_r(\bW)^{+}$ denotes the Moore--Penrose generalized inverse of $\bGam_r(\bW)=\bR_r\bSig_r(\bW)\bR_r'$.
\end{assumption}

Assumption~\ref{ass:nondeg} imposes no rank restriction on $\bW$.
For almost every $w$ with $\rank(w)=k$,
\[
\rank\bGam_r(w)
=
\rank\bR_r(w)
=
\binom{k+r-1}{r}
\]
by \eqref{eq:general-rank}.
Thus $\bGam_r(w)$ is nonsingular almost surely on the event
$\{\rank(\bW)=T\}$, and on that event
$\bGam_r(w)^+=\bGam_r(w)^{-1}$.
Positive definiteness of $\bSig_r(\bW)$ also implies
\[
\operatorname{col}\big(\bGam_r(\bW)\big)
=
\operatorname{col}\big(\bR_r(\bW)\big)
\quad\text{a.s.}
\]
Hence
\[
\bGam_r\bGam_r^+\bR_r=\bR_r
\quad\text{a.s.},
\]
so the generalized inverse acts as an inverse on the column space containing
the residual
\[
\bY^{[r]}-\bR_r\bdel^{[r]}
=
\bR_r\big(\bD^{[r]}-\bdel^{[r]}\big).
\]
Moreover, $\bGam_r^+$ is positive definite on
$\operatorname{col}(\bR_r(\bW))$.
Under the support condition of Theorem~\ref{thm:support}, for every
$\bc\neq\bzero$,
\[
\bc'\bB_r\bc
=
\E\big[(\bR_r\bc)'\bGam_r^+(\bR_r\bc)\big]
>0,
\]
so $\bB_r$ is nonsingular.  Nonsingularity of $\bB_r$ is therefore not an
additional identification restriction.  If, in addition, the integrability
conditions in the definition of admissibility hold for $\bGam_r^+$, then
$\bA_r^*=\bGam_r^+$ is an admissible weight for estimation.

\begin{theorem}[Oracle optimal weight and conditional-moment efficiency bound]\label{thm:efficient}
Suppose Assumptions~\ref{ass:model}, \ref{ass:moment}, and \ref{ass:nondeg} and the support conditions of Theorem~\ref{thm:support} hold.  For any admissible weight $\bA_r$ satisfying \eqref{eq:int1}--\eqref{eq:int2},
\begin{equation}\label{eq:effineq}
\bM_{r\bA}^{-1}\bV_{r\bA}\bM_{r\bA}^{-1}
\succeq
\bB_r^{-1}.
\end{equation}
Equality is attained by the oracle weight
\begin{equation}\label{eq:Astar}
\bA_r^*(\bW)=\bGam_r(\bW)^{+},
\end{equation}
which equals $\bGam_r(\bW)^{-1}$ almost surely on
$\{\rank(\bW)=T\}$.
If, in addition, the observations are i.i.d.\ as in Assumption~\ref{ass:sampling} and the oracle weight itself satisfies \eqref{eq:int1}--\eqref{eq:int2}, then for this choice,
\begin{equation}\label{eq:effclt}
\sqrt{N}\big(\widehat{\bdel}^{[r]}-\bdel^{[r]}\big)
\dto
N(\bzero,\bB_r^{-1}).
\end{equation}
Moreover, $\bB_r^{-1}$ is the Chamberlain (1987) semiparametric efficiency bound for $\bdel^{[r]}$ in the conditional moment model
\begin{equation}\label{eq:cmr}
\E\big[\bY^{[r]}-\bR_r(\bW)\bdel^{[r]}\mid\bW\big]=\bzero.
\end{equation}
\end{theorem}

\begin{remark}[Scope of the efficiency statement]\label{rem:effscope}
The efficiency claim in Theorem~\ref{thm:efficient} is relative to the
semiparametric model defined by the conditional-moment restriction
\eqref{eq:cmr}.
The structural random coefficient model may impose additional restrictions
on the conditional distribution of $\bY^{[r]}$ given $\bW$.
We therefore do not claim efficiency relative to all restrictions implied
by the structural model.
\end{remark}

The oracle calculation identifies the minimum asymptotic covariance within
the weighted minimum-distance class, and Chamberlain's result establishes
that the same covariance matrix is the semiparametric efficiency bound for
the conditional-moment model \eqref{eq:cmr}.
This is an oracle benchmark rather than a feasible efficiency result because
we do not develop an estimator of $\bGam_r(\bW)$.
The matrix $\bB_r$ also provides the information representation of the
support condition.

For a regressor history $w$, define
\begin{equation}\label{eq:Bcond}
\bB_r(w)
=
\bR_r(w)'\bGam_r(w)^{+}\bR_r(w),
\end{equation}
so that
\[
\bB_r=\E[\bB_r(\bW)].
\]
In the conditional-moment model \eqref{eq:cmr},
$\bB_r(w)$ is the conditional information matrix for
$\bdel^{[r]}$ at the regressor history $w$.

\begin{corollary}[Information and regressor rank]\label{cor:infogeom}
Suppose Assumption~\ref{ass:nondeg} holds.
\begin{enumerate}[label=(\roman*)]
\item \(\ker\bB_r=\bigcap_{w\in\cW}\ker\bR_r(w)=\cN_r\); through \(\bc\mapsto P_{\bc}\), this space corresponds to the degree-\(r\) homogeneous polynomials \(\cI_r(\cX)\) that vanish on \(\cX\).  Consequently, \(\bB_r\) is nonsingular if and only if the support condition of Theorem~\ref{thm:support} holds.
\item For every \(w\) at which \(\bSig_r(w)\) is positive definite,
\[
\ker\bB_r(w)=\ker\bR_r(w),
\qquad
\rank\bB_r(w)=\binom{\rank(w)+r-1}{r}.
\]
In particular, in the intercept specification, a regressor history with no within-unit regressor variation has \(\rank(w)=1\) and hence a rank-one conditional information matrix.  If \(T<q\), then \(\rank\bB_r(w)\le S_r<Q_r\), so the conditional information matrix is singular at every realization, while \(\bB_r\) is nonsingular whenever the support condition holds.
\end{enumerate}
\end{corollary}

\begin{remark}[Singular conditional information and regular estimation]
\label{rem:infogeom}
Corollary~\ref{cor:infogeom} gives an information interpretation of the
support condition for the conditional-moment model \eqref{eq:cmr}.
When $T<q$, the conditional information matrix $\bB_r(w)$ is singular at
every regressor history.
Its null space satisfies
\[
\ker\bB_r(w)=\ker\bR_r(w)
\]
and, through $\bc\mapsto P_{\bc}$, corresponds to the degree-$r$
homogeneous polynomials that vanish on $\rowsp(w)$.
These null spaces can differ across regressor histories.
The information matrix is nonsingular exactly when their
intersection is trivial.

Because Assumption~\ref{ass:nondeg} does not impose full row rank, this
argument also covers populations containing regressor histories of different
ranks.
For example, in the intercept-plus-scalar-regressor setting of
Corollary~\ref{cor:scalar}, a stayer has a rank-one conditional information
matrix, while a mover has conditional information matrix of rank $S_r$.
Both enter the matrix $\bB_r$.
When the support condition holds, Theorem~\ref{thm:asynorm} gives
$\sqrt N$-consistent and asymptotically normal estimation of the moment
vector, and the infeasible oracle estimator of
Theorem~\ref{thm:efficient} attains the conditional-moment efficiency bound
$\bB_r^{-1}$.

This differs from Graham and Powell (2012).
In their correlated random coefficient panel model, with the time dimension
equal to the number of random coefficients, the semiparametric information
for the average partial effect is singular even though the effect remains
identified through stayers and the behavior of conditional means near
singular regressor histories.
Estimation is therefore irregular and proceeds at a slower-than-$\sqrt N$
rate.
Here the conditional information matrix is singular at every regressor
history, but aggregation across histories can yield a nonsingular information matrix.

If $\bB_r$ is singular, then
$\ker\bB_r=\cN_r\neq\{\bzero\}$.
Linear functionals $\ba'\bdel^{[r]}$ with
$\ba\not\perp\cN_r$ are not identified under the maintained
conditional-moment restrictions by Theorem~\ref{thm:partialset}, and
Section~\ref{sec:partial} characterizes the resulting identified set.
We make no irregular-identification claim in this case.
As in Remark~\ref{rem:effscope}, these statements concern the
conditional-moment model \eqref{eq:cmr}, not all restrictions implied by the
structural random coefficient model.
\end{remark}

\begin{examplecont}{Example~\ref{ex:runningA} (continued: information)}
The two-history support in Example~\ref{ex:runningA} illustrates
Corollary~\ref{cor:infogeom}.
Suppose $\bSig_2(w)$ is positive definite at both support points.
Since both regressor histories have full row rank,
$\bGam_2(w)$ is positive definite and
$\bGam_2(w)^+=\bGam_2(w)^{-1}$ at those support points.
For $j=1,2$, \eqref{eq:Bcond} gives
\[
\rank \bB_2(\bW^{(j)})
=
\rank \bR_2(\bW^{(j)})
=
3<6.
\]
Thus each conditional information matrix has a three-dimensional null space.

Let
\[
\pi=\Pr\{\bW=\bW^{(1)}\}\in(0,1).
\]
The information matrix is
\[
\bB_2
=
\pi\bB_2(\bW^{(1)})
+
(1-\pi)\bB_2(\bW^{(2)}).
\]
Its rank is
\[
\rank(\bB_2)=5,
\qquad
\ker(\bB_2)
=
\operatorname{span}\{(0,0,0,0,1,0)'\}.
\]
The remaining null direction corresponds to the second moment
$\E[D_2D_3]$.

The three characterizations of the failure of second-order identification
therefore agree in this example.
The polynomial $x_2x_3$ vanishes on both row spaces,
$\E[D_2D_3]$ is the unique unidentified second moment, and its coefficient
direction spans the null space of the information matrix.
\end{examplecont}

\subsection{Exact recovery and the oracle estimator}\label{sec:exact}

When $q=T$ and $\bW$ is nonsingular,
\[
\bD=\bW^{-1}\bY,
\]
so the coefficient vector is recovered unit by unit.
The oracle estimator then has a direct interpretation in terms of averages
of the recovered coefficient products.

\begin{corollary}[Exact recovery when $q=T$]\label{cor:exactrecover}
Suppose $q=T$ and $\Pr\{\rank(\bW)=T\}=1$.  Then $\bR_r(\bW)$ is square and nonsingular and
\begin{equation}\label{eq:Drecover}
\bD_i^{[r]}=\bR_{r,i}^{-1}\bY_i^{[r]}.
\end{equation}
If, in addition, $\Var(\bD^{[r]}\mid\bW)=\bSig_r$ is constant and positive definite, and the oracle weight is used, then
\begin{equation}\label{eq:qTmean}
\widehat{\bdel}^{[r]}
=
\frac1N\sum_{i=1}^N\bD_i^{[r]},
\qquad
\sqrt{N}(\widehat{\bdel}^{[r]}-\bdel^{[r]})
\dto
N(\bzero,\bSig_r).
\end{equation}
If instead $\bSig_r(\bW)=\Var(\bD^{[r]}\mid\bW)$ is positive definite almost surely, $\E\|\bSig_r(\bW)\|<\infty$, $\E\|\bSig_r(\bW)^{-1}\|<\infty$, and the oracle weight satisfies \eqref{eq:int1}--\eqref{eq:int2}, so that the left-hand side below is the asymptotic covariance in \eqref{eq:effclt}, the oracle weight yields
\begin{equation}\label{eq:qTharmonic}
\bM_{r\bA^*}^{-1}\bV_{r\bA^*}\bM_{r\bA^*}^{-1}
=
\big(\E[\bSig_r(\bW)^{-1}]\big)^{-1}
\preceq
\E[\bSig_r(\bW)]
=
\Var(\bD^{[r]}),
\end{equation}
with equality if and only if $\bSig_r(\bW)$ is constant almost surely; the right-hand side is the asymptotic covariance of the unweighted average $\overline{\bD^{[r]}}=N^{-1}\sum_{i=1}^N\bD_i^{[r]}$.
\end{corollary}

Without conditional homoskedasticity of $\bD^{[r]}$, exact recovery does not
imply equal weighting of the recovered coefficient products.
The oracle estimator weights the recovered $\bD_i^{[r]}$ according to their
conditional covariance matrices.
The matrix harmonic--arithmetic mean inequality
\eqref{eq:qTharmonic} gives the resulting reduction in asymptotic covariance
relative to the unweighted average.
The two covariance matrices coincide if and only if
$\bSig_r(\bW)$ is constant almost surely.

\section{Joint estimation of moments through order \texorpdfstring{$r$}{r}}\label{sec:stacked}

Covariances and higher central co-moments depend jointly on raw moments of
several orders.
To estimate these moments jointly, define
\begin{equation}\label{eq:stack}
\bY^{[\le r]}
=
\begin{pmatrix}
\bY^{[1]}\\
\vdots\\
\bY^{[r]}
\end{pmatrix},
\qquad
\bD^{[\le r]}
=
\begin{pmatrix}
\bD^{[1]}\\
\vdots\\
\bD^{[r]}
\end{pmatrix},
\qquad
\bR_{\le r}
=
\diag(\bR_1,\ldots,\bR_r),
\end{equation}
so that
\begin{equation}\label{eq:stackrep}
\bY^{[\le r]}=\bR_{\le r}\bD^{[\le r]}.
\end{equation}
Let
\begin{equation}\label{eq:stackdelta}
\bdel^{[\le r]}=\E[\bD^{[\le r]}].
\end{equation}

\begin{lemma}[Dimensions and rank of the stacked matrix]\label{lem:dims}
The dimensions are
\begin{equation}\label{eq:dims}
S_{\le r}
=
\sum_{s=1}^r\binom{T+s-1}{s}
=
\binom{T+r}{r}-1,
\qquad
Q_{\le r}
=
\sum_{s=1}^r\binom{q+s-1}{s}
=
\binom{q+r}{r}-1.
\end{equation}
If $\Pr\{\rank(\bW)=T\}=1$, then
\begin{equation}\label{eq:stackrank}
\rank(\bR_{\le r})=S_{\le r}
\quad\text{a.s.}
\end{equation}
\end{lemma}

For the stacked system, replace Assumption~\ref{ass:moment} by the following
restriction.

\begin{assumption}[Moment homogeneity through order $r$]\label{ass:momentstack}
For $s=1,\ldots,r$,
\begin{equation}\label{eq:momentstack}
\E[\bD^{[s]}\mid\bW]=\bdel^{[s]}
\quad\text{a.s.}
\end{equation}
\end{assumption}

Let $\bA_{\le r}(\bW)$ be a known measurable admissible
$S_{\le r}\times S_{\le r}$ weight, with admissibility defined as in
Section~\ref{sec:identification} after replacing $\bR_r$ by $\bR_{\le r}$,
and define
\begin{equation}\label{eq:Mstack}
\bM_{\le r,\bA}
=
\E[\bR_{\le r}'\bA_{\le r}\bR_{\le r}].
\end{equation}
By Proposition~\ref{prop:toporder}, the order-$r$ support condition of
Theorem~\ref{thm:support} implies that $\bM_{\le r,\bA}$ is nonsingular for
every admissible stacked weight.
For asymptotic normality, impose the stacked analogues of
\eqref{eq:int1}--\eqref{eq:int2}.
Consistency of the covariance estimator based on the estimated moment
contributions additionally requires the stacked analogue of
\eqref{eq:int3}.

On the event that the sample Gram matrix is nonsingular, define
\begin{equation}\label{eq:stackest}
\widehat{\bdel}^{[\le r]}
=
\left(
\frac1N\sum_{i=1}^N\bR_{\le r,i}'\bA_{\le r,i}\bR_{\le r,i}
\right)^{-1}
\frac1N\sum_{i=1}^N\bR_{\le r,i}'\bA_{\le r,i}\bY_i^{[\le r]}.
\end{equation}
Under the conditions of Theorem~\ref{thm:stacked}, this event has probability
approaching one.

\begin{theorem}[Stacked moments]\label{thm:stacked}
Under Assumptions~\ref{ass:model}, \ref{ass:momentstack}, and
\ref{ass:sampling}, the order-$r$ support condition of
Theorem~\ref{thm:support}, and the stacked analogues of
\eqref{eq:int1}--\eqref{eq:int2},
\begin{equation}\label{eq:stackclt}
\sqrt{N}\big(\widehat{\bdel}^{[\le r]}-\bdel^{[\le r]}\big)
\dto
N\left(
\bzero,
\bM_{\le r,\bA}^{-1}
\bV_{\le r,\bA}
\bM_{\le r,\bA}^{-1}
\right),
\end{equation}
where
\begin{align}
\bGam_{\le r}(\bW)
&=
\Var(\bY^{[\le r]}\mid\bW),\label{eq:Gammastack}\\
\bV_{\le r,\bA}
&=
\E\big[
\bR_{\le r}'\bA_{\le r}\bGam_{\le r}\bA_{\le r}\bR_{\le r}
\big].\label{eq:Vstack}
\end{align}
Whenever $\bSig_{\le r}(\bW)=\Var(\bD^{[\le r]}\mid\bW)$ exists as a
finite matrix, $\bGam_{\le r}(\bW)=\bR_{\le r}\bSig_{\le r}(\bW)\bR_{\le r}'$.
If, in addition, the stacked analogue of \eqref{eq:int3} holds, the
sandwich covariance estimator obtained by replacing the order-$r$ objects in
Theorem~\ref{thm:robust} by their stacked counterparts is consistent.

If
$\bSig_{\le r}(\bW)=\Var(\bD^{[\le r]}\mid\bW)$ is positive definite
almost surely and
\begin{equation}\label{eq:Bstack}
\bB_{\le r}
=
\E[\bR_{\le r}'\bGam_{\le r}^{+}\bR_{\le r}]
\end{equation}
is finite, where $\bGam_{\le r}^{+}$ is the Moore--Penrose generalized
inverse of $\bGam_{\le r}(\bW)$, then $\bB_{\le r}$ is nonsingular under
the order-$r$ support condition.
Indeed, positive definiteness of $\bSig_{\le r}(\bW)$ implies that
$\bGam_{\le r}^{+}$ is positive definite on the column space of
$\bR_{\le r}(\bW)$, and Proposition~\ref{prop:toporder} together with
Theorem~\ref{thm:support} gives the corresponding common-null-space
condition.
The Chamberlain efficiency bound for $\bdel^{[\le r]}$ in the stacked
conditional-moment model is then
\begin{equation}\label{eq:stackbound}
\bB_{\le r}^{-1},
\end{equation}
and the oracle weight is
$\bA_{\le r}^*=\bGam_{\le r}^{+}$.
Application of the oracle estimator additionally requires the corresponding
integrability conditions.  Feasible implementation would also require an
estimator of $\bGam_{\le r}(\bW)$, which we do not develop here.
\end{theorem}

\begin{remark}[Joint versus separate estimation]\label{rem:stack}
Estimating each order separately remains consistent under the corresponding
single-order conditions.
Joint estimation permits the weight matrix $\bA_{\le r}$ to use covariance
between outcome-product vectors of different orders; a positive-definite
stacked weight need not be block diagonal.
In particular, the oracle weight uses the full conditional covariance matrix
$\bGam_{\le r}(\bW)$, including its off-diagonal blocks across moment
orders.
\end{remark}

\section{Mean and central co-moments}\label{sec:central}

The identification and estimation results above are stated in terms of raw
moments because these enter linearly in the higher-order moment
representation.
For interpretation, it is often more useful to report the mean, covariance,
and higher central co-moments.
Let
\begin{equation}\label{eq:mu}
\bmu=\bdel^{[1]}=\E[\bD],
\qquad
\bDt=\bD-\bmu,
\end{equation}
and, for $s\ge2$,
\begin{equation}\label{eq:gamma}
\bgam^{[s]}
=
\E[(\bD-\bmu)^{[s]}].
\end{equation}
Define
\begin{equation}\label{eq:theta}
\bth^{[\le r]}
=
\big(
\bmu',
(\bgam^{[2]})',
\ldots,
(\bgam^{[r]})'
\big)'.
\end{equation}
The dimension of $\bth^{[\le r]}$ is $Q_{\le r}$.

For a multi-index $\ba\in\mathbb{N}_0^q$ with $|\ba|\ge2$, write
$\gamma_{\ba}=\E[(\bD-\bmu)^{\ba}]$ and
$\delta_{\bb}=\E[\bD^{\bb}]$, with $\delta_{\bzero}=1$.
Let
\[
\binom{\ba}{\bb}=\prod_{j=1}^q\binom{a_j}{b_j},
\qquad
c(\ba,\bb)=(-1)^{|\ba-\bb|}\binom{\ba}{\bb}.
\]

\begin{lemma}[Raw--central conversion]\label{lem:conversion}
For $|\ba|\ge2$,
\begin{equation}\label{eq:rawcentral}
\gamma_{\ba}
=
\sum_{\bb:\,\bb\le\ba}
c(\ba,\bb)\delta_{\bb}\bmu^{\ba-\bb}.
\end{equation}
Conversely,
\begin{equation}\label{eq:centralraw}
\delta_{\ba}
=
\sum_{\bb:\,\bb\le\ba}
\binom{\ba}{\bb}\gamma_{\bb}\bmu^{\ba-\bb},
\end{equation}
where $\gamma_{\bzero}=1$ and $\gamma_{e_j}=0$ for the coordinate vectors $e_j$.
\end{lemma}

To construct the Jacobian of the transformation from raw to central moments,
it is convenient to collect the terms in \eqref{eq:rawcentral} involving
raw moments of orders zero and one.

\begin{lemma}[Collected form]\label{lem:collected}
For $|\ba|\ge2$,
\begin{equation}\label{eq:collected}
\gamma_{\ba}
=
(-1)^{|\ba|}(1-|\ba|)\bmu^{\ba}
+
\sum_{\substack{\bb:\,\bb\le\ba\\|\bb|>1}}
c(\ba,\bb)\delta_{\bb}\bmu^{\ba-\bb}.
\end{equation}
\end{lemma}

\begin{lemma}[Jacobian entries]\label{lem:derivs}
For $|\ba|\ge2$ and $k=1,\ldots,q$,
\begin{align}
\frac{\partial\gamma_{\ba}}{\partial\mu_k}
&=
1\{\ba\ge e_k\}a_k(-1)^{|\ba|}(1-|\ba|)\bmu^{\ba-e_k}
\nonumber\\
&\quad+
\sum_{\substack{\bb:\,\bb\le\ba-e_k\\|\bb|>1}}
(a_k-b_k)c(\ba,\bb)\delta_{\bb}\bmu^{\ba-\bb-e_k},
\label{eq:dgdm}\\
\frac{\partial\gamma_{\ba}}{\partial\delta_{\bb}}
&=
1\{\ba\ge\bb\}c(\ba,\bb)\bmu^{\ba-\bb},
\qquad |\bb|>1.
\label{eq:dgdd}
\end{align}
\end{lemma}

Let
\begin{equation}\label{eq:G}
\bG(\bdel^{[\le r]})
=
\frac{\partial\bth^{[\le r]}}
{\partial(\bdel^{[\le r]})'}.
\end{equation}

\begin{lemma}[Triangular Jacobian and identification]\label{lem:jacobian}
Order the raw and central moment vectors by total degree.  Then $\bG$ is block lower triangular with diagonal blocks
\begin{equation}\label{eq:Gdiag}
\frac{\partial\bmu}{\partial(\bdel^{[1]})'}=\bI_q,
\qquad
\frac{\partial\bgam^{[s]}}{\partial(\bdel^{[s]})'}=\bI_{Q_s},
\qquad s=2,\ldots,r.
\end{equation}
Hence $\det(\bG)=1$ and $\bG$ is nonsingular everywhere.  The map from raw moments $\bdel^{[\le r]}$ to $(\bmu,\bgam^{[2]},\ldots,\bgam^{[r]})$ is therefore a polynomial bijection with polynomial inverse.
\end{lemma}

Thus identification of the raw moments through order $r$ is equivalent to
identification of the mean and central co-moments through the same order.
Inference follows by applying the delta method to the stacked raw-moment
estimator.

\begin{theorem}[Inference for mean and central co-moments]\label{thm:delta}
Let $\widehat\bth^{[\le r]}$ be obtained by applying \eqref{eq:rawcentral} to $\widehat\bdel^{[\le r]}$.  Under the assumptions of Theorem~\ref{thm:stacked},
\begin{equation}\label{eq:deltaclt}
\sqrt N\big(\widehat\bth^{[\le r]}-\bth^{[\le r]}\big)
\dto
N\left(
\bzero,
\bG
\bM_{\le r,\bA}^{-1}
\bV_{\le r,\bA}
\bM_{\le r,\bA}^{-1}
\bG'
\right),
\end{equation}
where $\bG=\bG(\bdel^{[\le r]})$.  Moreover,
\begin{equation}\label{eq:Ghat}
\widehat\bG
=
\bG(\widehat\bdel^{[\le r]})
\pto
\bG.
\end{equation}
If, in addition, the stacked analogue of \eqref{eq:int3} holds, replacing $\bG$, $\bM_{\le r,\bA}$, and $\bV_{\le r,\bA}$ by their sample analogues yields a consistent covariance estimator.
\end{theorem}

\begin{remark}[Two-step centering]\label{rem:twostep}
The identity
\[
\bY-\bW\bmu
=
\bW(\bD-\bmu)
\]
suggests an alternative two-step estimator.
One could first estimate $\bmu$, form the centered outcomes
$\bY_i-\bW_i\widehat\bmu$, and then apply the moment representation to
estimate the central moments.
The stacked raw-moment estimator avoids a separate centering step, and the
Jacobian in Theorem~\ref{thm:delta} accounts directly for estimation
uncertainty in $\widehat\bmu$ and for covariance across moment orders.
\end{remark}

\section{Failure of the support condition: nonidentification and sharp identified sets}\label{sec:partial}

Suppose the support condition of Theorem~\ref{thm:support} fails.
Then the conditional-moment map is not injective and
$\cN_{\le r}$ contains nonzero directions.
Under the maintained conditional-moment restrictions, the model therefore
admits observationally equivalent structures with different coefficient
moments.
The results in this section serve two purposes.
They establish that failure of the rank condition implies
nonidentification in the maintained model, characterize the sharp
identified set for the coefficient moments, and determine which linear
functionals remain identified when the condition fails.

Maintain Assumptions~\ref{ass:model} and \ref{ass:momentstack}.
The conditional moments used for identification are
\begin{equation}\label{eq:info}
m_s(w)=\E[\bY^{[s]}\mid\bW=w]=\bR_s(w)\bdel_0^{[s]},
\qquad s=1,\ldots,r,\quad F_W\text{-a.e. }w,
\end{equation}
where $\bdel_0^{[\le r]}$ is the true moment vector.
A candidate $\bdel^{[\le r]}\in\mathbb R^{Q_{\le r}}$ is
\emph{observationally equivalent to $\bdel_0^{[\le r]}$ relative to
\eqref{eq:info}} if there exists a conditional distribution of $\bD$ given
$\bW$ that satisfies Assumption~\ref{ass:momentstack}, has moment vector
$\bdel^{[\le r]}$, and reproduces $m_s(w)$ for every $s=1,\ldots,r$ and
$F_W$-almost every $w$.\footnote{For $r\ge3$, the full joint distribution of
$(\bY,\bW)$ may impose restrictions beyond the conditional moments
\eqref{eq:info}; Corollary~\ref{cor:sharpdata} shows that it imposes no
additional restrictions when $r=2$.  Exploiting such information at higher
orders would require distributional inversion, which we do not pursue.
Conditional moments above order $r$ may also restrict moments through order
$r$ through truncated-moment extension conditions.  Unless otherwise
stated, the identified-set results in this section are sharp with respect
to the information in \eqref{eq:info}, which is also the information used
by the estimators in Sections~\ref{sec:estimation}--\ref{sec:stacked}.}

By block diagonality of $\bR_{\le r}$, the unidentified moment space is
\begin{equation}\label{eq:Nstack}
\cN_{\le r}
=
\cN_1\times\cdots\times\cN_r,
\qquad
\cN_s=\{\bc:\bR_s(w)\bc=\bzero\ \forall w\in\cW\},
\end{equation}
with each $\cN_s$ isomorphic to $\cI_s(\cX)$ by
Theorem~\ref{thm:support}.
For finite support, $\cN_s$ is the null space of the matrix obtained by
stacking $\bR_s(w)$ over the support points, as described in
Appendix~\ref{app:geometry}.
For general support, the equivalent characterization follows
from the Gram-matrix condition in Theorem~\ref{thm:support}.

Define the truncated moment set
\begin{equation}\label{eq:momentcone}
\cM_{\le r}
=
\Big\{
\mathbf{m}\in\mathbb R^{Q_{\le r}}:
\exists\ \text{a distribution } F \text{ on }\mathbb R^q
\text{ with }
{\textstyle\int} \bx^{[s]}\,dF=\mathbf{m}^{[s]},\ s=1,\ldots,r
\Big\}.
\end{equation}
The set $\cM_{\le r}$ is convex and imposes the requirement that a candidate
vector be the collection of moments through order $r$ of some probability
distribution on $\mathbb R^q$.

\begin{theorem}[Model-level nonidentification and sharp identified set]\label{thm:partialset}
Under Assumptions~\ref{ass:model} and \ref{ass:momentstack}, the identified set for $\bdel^{[\le r]}$ relative to the information \eqref{eq:info} is 
\begin{equation}\label{eq:idset}
\Theta_I
=
\big(\bdel_0^{[\le r]}+\cN_{\le r}\big)\cap\cM_{\le r},
\end{equation}
a convex set.
Every element of $\Theta_I$ is attained by a structure in which $\bD$ is
independent of $\bW$.
Consequently, strengthening Assumption~\ref{ass:momentstack} to full
independence does not shrink the identified set relative to
\eqref{eq:info}.
Such a structure reproduces the conditional moments \eqref{eq:info} but
need not reproduce the conditional distribution of $\bY$ given $\bW$.
At $r=2$, Corollary~\ref{cor:sharpdata} constructs a conditionally
heterogeneous structure that also reproduces that distribution.

If $\ba\perp\cN_{\le r}$, the linear functional
$\ba'\bdel^{[\le r]}$ has the same value throughout every affine class
$\bdel_0^{[\le r]}+\cN_{\le r}$ and is therefore identified under the
maintained conditional-moment restrictions.
Conversely, if $\ba\not\perp\cN_{\le r}$, the maintained model contains
structures generating the same conditional moments and different
values of the functional, so the functional is not identified as a
parameter of the maintained model.
The dimension of the space of linear functionals identified as parameters
of the maintained conditional-moment model is given by
\eqref{eq:hilbertstack} in Appendix~\ref{app:geometry}.
\end{theorem}

The two restrictions in \eqref{eq:idset} have distinct roles.
The affine space
$\bdel_0^{[\le r]}+\cN_{\le r}$ contains the moment vectors that cannot be
distinguished by the conditional moments \eqref{eq:info}.
Intersecting this space with $\cM_{\le r}$ removes vectors that are not
moments of any probability distribution.

\begin{remark}[Identification at a distribution and in the model]
\label{rem:fixedvsmodel}
Theorem~\ref{thm:partialset} involves two notions of identification.
For a fixed observed conditional-moment function, generated by
$\bdel_0^{[\le r]}$, a linear functional $\ba'\bdel^{[\le r]}$ is point
identified at that distribution if and only if it is constant on
$\Theta_I$.
The functional claims in Theorem~\ref{thm:partialset} concern the
maintained model: $\ba\perp\cN_{\le r}$ makes the functional constant on
the identified set for every $\bdel_0^{[\le r]}\in\cM_{\le r}$, whereas
$\ba\not\perp\cN_{\le r}$ produces some
$\bdel_0^{[\le r]}\in\cM_{\le r}$ at which it is not constant.
The two notions coincide away from the boundary of the truncated moment
set: if $\bdel_0^{[\le r]}$ lies in the interior of $\cM_{\le r}$, then,
for every fixed $\mathbf n\in\cN_{\le r}$,
$\bdel_0^{[\le r]}\pm\varepsilon\mathbf n\in\cM_{\le r}$ for all
sufficiently small $\varepsilon>0$, so
$\Theta_I$ is a singleton if and only if $\cN_{\le r}=\{\bzero\}$.
Collapse of $\Theta_I$ to a singleton when
$\cN_{\le r}\neq\{\bzero\}$ therefore requires
$\bdel_0^{[\le r]}$ to lie on the boundary of $\cM_{\le r}$; at $r=2$,
by Theorem~\ref{thm:spectra}, it requires a singular coefficient
covariance matrix.
The zero-residual-variance cases in Example~\ref{ex:runningA} and
Proposition~\ref{prop:binarybounds} are of this type.
Boundary position is necessary for collapse, not sufficient.
\end{remark}

At second order, moment feasibility has an exact positive-semidefinite
representation.
Write $\operatorname{mat}(\bdel^{[2]})$ for the symmetric $q\times q$
matrix with entries $\E[D_jD_k]$, and define
\begin{equation}\label{eq:momentmatrix}
M(\bdel^{[\le2]})
=
\begin{pmatrix}
1 & (\bdel^{[1]})'\\
\bdel^{[1]} & \operatorname{mat}(\bdel^{[2]})
\end{pmatrix},
\end{equation}
which is affine in $\bdel^{[\le2]}$.

\begin{theorem}[Second-order positive-semidefinite characterization]\label{thm:spectra}
For $r=2$, $\cM_{\le2}=\{\bdel^{[\le2]}:M(\bdel^{[\le2]})\succeq0\}$, so
\begin{equation}\label{eq:spectra}
\Theta_I
=
\big\{\bdel^{[\le2]}\in\bdel_0^{[\le2]}+\cN_{\le2}:M(\bdel^{[\le2]})\succeq0\big\},
\end{equation}
the intersection of an affine subspace with a linear matrix inequality, hence a spectrahedron.
The set $\Theta_I$ is closed, being the intersection of an affine subspace
with the preimage of the closed positive-semidefinite cone under the affine
map
$\bdel^{[\le2]}\mapsto M(\bdel^{[\le2]})$.
Sharp lower and upper bounds on any linear functional
$\ba'\bdel^{[\le2]}$ are therefore the infimum and supremum of a pair of
semidefinite programs.
Every feasible moment vector is attained by a possibly singular Gaussian
distribution.
Corollary~\ref{cor:compact} gives the exact condition under which
$\Theta_I$ is compact; under that condition, every SDP optimum is attained
and the corresponding endpoint is attained by such a Gaussian structure.
Equivalently, by the Schur complement, feasibility is
$\operatorname{mat}(\bdel^{[2]})
-\bdel^{[1]}(\bdel^{[1]})'\succeq0$:
the candidate second moments must imply a positive-semidefinite covariance
matrix.
\end{theorem}

At second order, the identified set in
Theorem~\ref{thm:spectra} is also sharp relative to the full joint
distribution of $(\bY,\bW)$.
For $w\in\cW$ with $\rank(w)=k$, let $\mathbf{O}_w$ be a $q\times k$
matrix with orthonormal columns spanning $\rowsp(w)$ and
$\mathbf{N}_w$ a $q\times(q-k)$ matrix with orthonormal columns spanning
$\ker w$, so that
$\mathbf{O}_w\mathbf{O}_w'+\mathbf{N}_w\mathbf{N}_w'=\bI_q$,
$w\mathbf{N}_w=\bzero$, and
$w=w\mathbf{O}_w\mathbf{O}_w'$ with $w\mathbf{O}_w$ of full column rank
$k$.

\begin{corollary}[Sharpness relative to the data at second order]\label{cor:sharpdata}
Let $r=2$ and maintain Assumptions~\ref{ass:model} and \ref{ass:momentstack}.  For every $\bdel^{[\le2]}\in\Theta_I$ there is a conditional distribution
of $\bD$ given $\bW$ such that
(i) $\E[\bD^{[\le2]}\mid\bW=w]=\bdel^{[\le2]}$ for
$F_W$-almost every $w$, and
(ii) the joint distribution of $(\bW,\bW\bD)$ is the observed joint
distribution of $(\bW,\bY)$.
Consequently the identified set for $\bdel^{[\le2]}$ relative to the full
joint distribution of $(\bY,\bW)$ is the spectrahedron
\eqref{eq:spectra}, and no feature of the conditional distribution of
$\bY$ given $\bW$ beyond its first two moments tightens the bounds of
Theorem~\ref{thm:spectra}.
\end{corollary}

Conditional on $\bW=w$, the construction underlying
Corollary~\ref{cor:sharpdata} writes $\bD$ as an affine function of $\bY$
plus an auxiliary Gaussian component in $\ker w$ that is independent of
$\bY$.
At second order, positive semidefiniteness of the candidate coefficient
covariance is exactly the condition needed to choose the covariance of this
auxiliary component while preserving the observed conditional distribution
of $\bY$.
For $r\ge3$, fixing the conditional distribution of $\bY$ also restricts
higher cross-moments between the row-space and null-space components of
$\bD$.
Those restrictions are not contained in $\cM_{\le r}$, which only imposes
moment feasibility on $\bD$ itself.
This is why the full-data sharpness result stops at second order.

\begin{corollary}[Compactness and attainment at second order]\label{cor:compact}
Let $r=2$.  The set $\Theta_I$ in \eqref{eq:spectra} is compact if and only if $\cN_1=\{\bzero\}$, equivalently if and only if $\cX$ spans $\mathbb R^q$, equivalently, when $\E\|\bW\|^2<\infty$, if and only if $\E[\bW'\bW]$ is nonsingular.
In that case every linear functional $\ba'\bdel^{[\le2]}$ attains its
infimum and supremum on $\Theta_I$, each at a possibly singular Gaussian
structure, and the semidefinite-program values in
Theorem~\ref{thm:spectra} are finite.
If $\cN_1\neq\{\bzero\}$, then for every nonzero $\bv$ orthogonal to
$\cX$, $\Var(\bv'\bD)$ is unbounded above on $\Theta_I$.
\end{corollary}

Thus a coefficient direction orthogonal to every row space generated by the
regressor support has neither an identified mean nor a bounded variance.
When $\E\|\bW\|^2<\infty$, nonsingularity of $\E[\bW'\bW]$ therefore
determines whether the second-order identified set is compact before the
second-order semidefinite program is solved.

\begin{examplecont}{Example~\ref{ex:runningA} (continued: sharp partial identification)}
The conditional-moment equations leave $m_{23}$ unrestricted along one
linear direction.
The positive-semidefinite restriction in
Theorem~\ref{thm:spectra} restricts the admissible values along that
direction.

Since the first moments are point identified, write
\[
\bSig
=
\Var(\bD)
=
\begin{pmatrix}
\sigma_{11} & \sigma_{12} & \sigma_{13}\\
\sigma_{12} & \sigma_{22} & \sigma_{23}\\
\sigma_{13} & \sigma_{23} & \sigma_{33}
\end{pmatrix}.
\]
The two observed systems identify every entry of $\bSig$ except
$\sigma_{23}$.
Because
\[
m_{23}
=
\mu_2\mu_3+\sigma_{23},
\]
bounding $m_{23}$ is equivalent to bounding $\sigma_{23}$.

Suppose for simplicity that $\sigma_{11}>0$.
Positive semidefiniteness of $\bSig$ is equivalent to positive
semidefiniteness of the Schur complement
\[
\begin{pmatrix}
\sigma_{22}-\dfrac{\sigma_{12}^2}{\sigma_{11}}
&
\sigma_{23}-\dfrac{\sigma_{12}\sigma_{13}}{\sigma_{11}}
\\[10pt]
\sigma_{23}-\dfrac{\sigma_{12}\sigma_{13}}{\sigma_{11}}
&
\sigma_{33}-\dfrac{\sigma_{13}^2}{\sigma_{11}}
\end{pmatrix}.
\]
Hence the sharp identified interval is
\begin{align}
\sigma_{23}
\in
\Bigg[
&
\frac{\sigma_{12}\sigma_{13}}{\sigma_{11}}
-
\sqrt{
\left(
\sigma_{22}-\frac{\sigma_{12}^2}{\sigma_{11}}
\right)
\left(
\sigma_{33}-\frac{\sigma_{13}^2}{\sigma_{11}}
\right)
},
\nonumber\\
&
\frac{\sigma_{12}\sigma_{13}}{\sigma_{11}}
+
\sqrt{
\left(
\sigma_{22}-\frac{\sigma_{12}^2}{\sigma_{11}}
\right)
\left(
\sigma_{33}-\frac{\sigma_{13}^2}{\sigma_{11}}
\right)
}
\Bigg].
\label{eq:running-cov-bound}
\end{align}

The interval is nondegenerate whenever both residual variances are strictly
positive.
If either residual variance is zero, the interval collapses to a singleton.
Thus moment feasibility can produce point identification at such a degenerate observed distribution even though the conditional-moment map is
not injective under the maintained model (Remark~\ref{rem:fixedvsmodel}).

For interpretation, define the linear residuals
\[
U_2
=
D_2-\frac{\sigma_{12}}{\sigma_{11}}D_1,
\qquad
U_3
=
D_3-\frac{\sigma_{13}}{\sigma_{11}}D_1.
\]
Then
\[
\Cov(U_2,U_3)
=
\sigma_{23}
-
\frac{\sigma_{12}\sigma_{13}}{\sigma_{11}},
\qquad
\Var(U_2)
=
\sigma_{22}-\frac{\sigma_{12}^2}{\sigma_{11}},
\qquad
\Var(U_3)
=
\sigma_{33}-\frac{\sigma_{13}^2}{\sigma_{11}}.
\]
The conditional-moment equations identify the two residual variances but
not their covariance.
The interval in \eqref{eq:running-cov-bound} is therefore exactly the
Cauchy--Schwarz restriction
\[
|\Cov(U_2,U_3)|
\le
\sqrt{\Var(U_2)\Var(U_3)}.
\]

The interval is sharp.
Every value satisfying \eqref{eq:running-cov-bound} produces a
positive-semidefinite covariance matrix $\bSig$.
By Theorem~\ref{thm:spectra}, a possibly singular Gaussian distribution
with mean $\bmu$ and covariance $\bSig$ attains those moments.
Taking $\bD$ independent of $\bW$ with this distribution reproduces the
observed conditional first and second moments.
Corollary~\ref{cor:sharpdata} further provides a conditional distribution
of $\bD$ given $\bW$ that reproduces the full observed conditional
distribution of $\bY$ for every value in
\eqref{eq:running-cov-bound}.
The interval is bounded, as Corollary~\ref{cor:compact} implies, because
$L_1\cup L_2$ spans $\mathbb R^3$.

Thus the conditional-moment equations determine the affine unidentified
direction, while positive-semidefinite moment feasibility determines the
sharp subset of that direction.
\end{examplecont}

\begin{examplecont}{Example~\ref{ex:runningA} (concluded: restoring identification)}
Enlarge the regressor support by adding
\[
\bW^{(3)}
=
\begin{pmatrix}
0&1&0\\
0&0&1
\end{pmatrix},
\]
with positive probability.
Its row space is
\[
L_3
=
\rowsp(\bW^{(3)})
=
\{\bx\in\mathbb R^3:x_1=0\}.
\]
Under this regressor history,
$\bY=(D_2,D_3)'$, so
$m_{23}=\E[D_2D_3]$ is directly identified.

The support characterization gives the same conclusion.
A homogeneous quadratic that vanishes on
$L_1=\{x_3=0\}$ must be divisible by $x_3$.
If it also vanishes on $L_2=\{x_2=0\}$, it must be divisible by
$x_2x_3$.
Vanishing additionally on $L_3=\{x_1=0\}$ would require divisibility by
$x_1x_2x_3$, which is impossible for a nonzero quadratic.
Hence
\[
\cI_2(L_1\cup L_2\cup L_3)=\{0\},
\qquad
\cN_2=\{0\}.
\]
Equivalently,
\[
\rank
\begin{pmatrix}
\bR_2(\bW^{(1)})\\
\bR_2(\bW^{(2)})\\
\bR_2(\bW^{(3)})
\end{pmatrix}
=
6.
\]
This is the $q=3$, $T=2$, $r=2$ case of
Theorem~\ref{thm:hyper}: three distinct hyperplanes are necessary and
sufficient for the second-order support condition.

Under conditional nondegeneracy, the same conclusion follows from the
information matrix.
Each $\bB_2(\bW^{(j)})$ is singular with rank three, but
\[
\ker(\bB_2)
=
\bigcap_{j=1}^3
\ker\bR_2(\bW^{(j)})
=
\{0\},
\]
so $\bB_2$ is nonsingular.

The interval in \eqref{eq:running-cov-bound} therefore collapses to the
value of $\sigma_{23}$ implied by the third regressor history.
Each individual system remains underdetermined, but the three row spaces
jointly identify the complete second-moment matrix of $\bD$.
\end{examplecont}

\begin{proposition}[Binary regressor: closed-form bounds]\label{prop:binarybounds}
Let $T=1$, $\bW=(1,X)$ with $X\in\{0,1\}$, and $r=2$, and write
$s_x^2=\Var(Y\mid X=x)$, $x\in\{0,1\}$, for the observed conditional
variances.
Then $\bdel^{[1]}$ is point identified,
$\cN_2=\operatorname{span}\{(0,1,-2)'\}$ in the ordering
$(D_0^2,D_0D_1,D_1^2)$, and the identified set of
Theorem~\ref{thm:spectra} projects onto
\begin{equation}\label{eq:binaryinterval}
\Var(D_1)\in\big[(s_1-s_0)^2,\ (s_1+s_0)^2\big],
\qquad
\operatorname{sd}(D_1)\in\big[\,|s_1-s_0|,\ s_1+s_0\big],
\end{equation}
with every value attained.
These are the bounds in Proposition~2.1 of Hermann and Holzmann (2025),
obtained here from the positive-semidefinite characterization.
Their Section~2.2, equations (2.9)--(2.11), considers the corresponding
problem when one binary regressor appears among other, more richly
supported regressors; the resulting higher-dimensional affine restriction
leads to a semidefinite program.
Their Proposition~2.5 characterizes when the upper bound of that program
for the slope variance is strictly positive.
Corresponding closed-form bounds follow for $\Cov(D_0,D_1)$.
\end{proposition}

Corollary~\ref{cor:scalar} gives the corresponding panel result.
With $T=2$, positive probability of movement,
$\Pr\{X_1\neq X_2\}>0$, restores the rank condition and point
identifies moments of every fixed order.
More generally, if $\cN_{\le r}=\{\bzero\}$, the moments through order $r$
are point identified.
If $\cN_{\le r}\neq\{\bzero\}$, they are not jointly point identified under
the maintained conditional-moment restrictions, although particular linear
functionals orthogonal to $\cN_{\le r}$ remain identified.
For a given observed conditional-moment function, the sharp identified set
is \eqref{eq:idset}. When an identified residual variance is zero, the interval collapses to a point; see Example~\ref{ex:runningA} and Proposition~\ref{prop:binarybounds}.

\begin{remark}[Higher orders]\label{rem:higherorders}
For $r\ge3$, membership in $\cM_{\le r}$ is a truncated multivariate
moment problem.
Positive semidefiniteness of the moment matrix of order
$\lfloor r/2\rfloor$ formed from the available moments is necessary, and
therefore so is positive semidefiniteness of its principal truncations.
When restrictions are imposed on the support of $\bD$, the corresponding
localizing-matrix conditions are also necessary insofar as they involve only
moments of order at most $r$.
At higher orders these semidefinite restrictions generally give convex outer
approximations rather than a complete characterization of
$\cM_{\le r}$.
Flat-extension conditions can certify representability when a suitable
positive extension exists (Curto and Fialkow, 1998).
If the coefficients are restricted to a compact basic semialgebraic set and
the standard Archimedean conditions hold, Lasserre's (2001) moment--sum-of-squares hierarchy provides convergent
semidefinite bounds.
We use these results only as computational characterizations at higher
orders; the second-order characterization needed for means, variances, and
covariances is exact.
\end{remark}

\begin{remark}[Computation and estimation]\label{rem:picompute}
For finite regressor support, $\cN_s$ is the null space of the matrix
obtained by stacking $\bR_s(w)$ over the support points, so the unidentified
moment space is obtained by a finite-dimensional null-space calculation.
For general support, $\cN_s$ can instead be characterized through the
corresponding Gram matrix and estimated from its sample analogue.
At second order, the bounds in Theorem~\ref{thm:spectra} are values of
semidefinite programs; replacing the conditional moments by
estimators yields an estimated second-order identified set.
Inference for the identified set and for the resulting SDP value functions
requires additional methods for partially identified models, which we do not
develop here.
\end{remark}

\begin{remark}[Related partial-identification results]\label{rem:pilit}
Theorem~\ref{thm:partialset} combines the affine restrictions implied by the
conditional-moment equations with feasibility restrictions from the
truncated moment problem.

At second order, the resulting positive-semidefinite completion problem is
related to statistical matching, where covariances between variables that
are never jointly observed are bounded by positive-definite completion
restrictions; see, e.g., Ahfock et al. (2016) and the matrix-completion
results of Grone et al. (1984).
Dobronyi et al. (2026) use truncated-moment restrictions and semidefinite
programming in dynamic panel logit models.
Their affine restrictions arise from the model's functional form, whereas
here they arise from the common null space generated by the regressor
support; at second order, the characterization here is exact without the
general truncated-moment machinery.
Positive semidefiniteness of latent moment matrices also appears in the
errors-in-variables analysis of Klepper and Leamer (1984) and Leamer
(1987).
The model also belongs to the broader class of latent-variable moment
models studied by Schennach (2014), with the latent coefficient vector
entering linearly after applying the higher-order moment representation.

Because $\Theta_I$ is convex, maximizing a linear functional over it
evaluates its support function.
Beresteanu, Molchanov, and Molinari (2011) develop support-function methods
for sharp convex identification regions.
See Molinari (2020) for a broader review.
Freyberger and Horowitz (2015) provide a linear-programming analogue for
conditional means with discrete instruments, combining finitely many affine
moment restrictions with shape restrictions.
Stoye (2010) analyzes partially identified spread parameters, which is
relevant here because variance is a leading bounded functional.

Within random coefficient models, Lee (2026) obtains bounds under
predetermined regressors.
Gaillac and Gautier (2021) instead obtain point identification of the full
coefficient distribution under limited regressor variation by imposing
quasi-analytic restrictions on the coefficient class.
The latter approach restricts the coefficient distribution to restore point
identification; here the coefficient class is left unrestricted.
\end{remark}

\section{Discussion}\label{sec:discussion}

The support theorem gives a primitive characterization of the rank condition.
Under the maintained conditional-moment restrictions, the order-$r$
coefficient moments are identified if and only if the matrices
$\bR_r(w)$ have no nonzero common null direction over the support of
$\bW$.
There is no separate pair of conditions corresponding to
``within'' and ``between'' rank, of the kind the within and between
decompositions of panel variation might suggest.
Each regressor history contributes restrictions through $\bR_r(w)$, and
identification depends on whether these restrictions jointly have full
column rank.
The row-space characterization expresses this condition directly in terms
of the original regressors.
Full row rank of a particular regressor history maximizes the number of
order-$r$ restrictions it contributes, but is neither necessary nor
sufficient for identification.

The same matrices provide diagnostics for the empirical strength of the
rank condition.
With finite regressor support, one can stack
$\bR_r(w_j)$ over support points and examine the rank and singular values
of the resulting matrix.
With continuously distributed regressors, the sample Gram matrix
\[
\frac1N\sum_{i=1}^N
\bR_{r,i}'\bA_{r,i}\bR_{r,i}
\]
provides the corresponding sample measure of conditioning.
Because $Q_r$ increases rapidly with $r$, the smallest eigenvalue can
become small at higher orders even when the rank condition
continues to hold.
The feasible order of the moment analysis may therefore be limited by weak
conditioning of the sample analogue of the moment equations.
These singular-value and eigenvalue diagnostics depend on the normalization
of the regressors and moment coordinates and should therefore be compared
only under a fixed scaling.
At order one, under identity weighting, the Gram matrix is $\E[\bW'\bW]$, when finite.
By Corollary~\ref{cor:compact}, its nonsingularity also determines whether
the second-order identified set is compact; the smallest eigenvalue of its
sample analogue provides a natural empirical diagnostic of this condition.

Panel observations change both the number of restrictions contributed by each
regressor history and the pattern of overlap among those restrictions.
When $T=1$, each history contributes one order-$r$ restriction.
A full-row-rank $T$-equation history contributes
\[
S_r=\binom{T+r-1}{r}
\]
independent restrictions.
Restrictions from different histories need not be distinct:
Proposition~\ref{prop:twopoint} gives the exact number duplicated by the
intersection of two full-row-rank row spaces.
In the intercept specification, if there is no within-unit regressor variation,
the identification problem reduces to the corresponding cross-sectional problem.
When $T=q$, a single full-rank regressor history identifies moments of every
fixed order.

The maintained moment-homogeneity restriction can be relaxed without
changing the higher-order moment representation.
Remark~\ref{rem:crc-relax} allows the conditional coefficient moments to
depend on known functions of $\bW$ through a finite-dimensional
specification.
Identification is then determined by the rank of the expanded
conditional-moment system, which may remain deficient even under rich
regressor support.
Remark~\ref{rem:control-relax} instead imposes moment homogeneity
conditional on a control variable, in which case the support condition is
applied conditionally on the control.
The corresponding estimation problems are not developed here.

Other extensions concern inference for the partially identified sets in
Section~\ref{sec:partial}, imposing moment-feasibility restrictions in
estimation even when the moment vector is point identified, fixed-design
asymptotics for chosen regressor histories as in
Remark~\ref{rem:fixeddesign}, and support conditions for more general
collections of row spaces.
These questions build directly on the rank and identified-set
characterizations developed above.

\section{Conclusion}\label{sec:conclusion}

We study identification of moments of a random coefficient vector
when each unit provides fewer equations than coefficients.  At each fixed
order, the higher-order moment representation converts the problem into a
finite-dimensional conditional linear system.  Under moment homogeneity,
the coefficient moments are identified if and only if the matrices
$\bR_r(w)$ have no nonzero common null direction over the regressor support.
Equivalently, no nonzero homogeneous polynomial of degree $r$ vanishes on
every row space generated by the support of the regressors.

The panel structure determines how many restrictions each regressor history
contributes and how those restrictions overlap across histories.  A
full-row-rank $T$-equation history contributes
$\binom{T+r-1}{r}$ independent order-$r$ restrictions, while intersections
of row spaces determine their redundancy.  Unit-level recovery of the
random coefficients is therefore not required for identification of their
moments.

When the support condition fails, its common null space determines the
unidentified moment directions, and moment feasibility determines the
corresponding identified set.  At second order, the identified set is a
spectrahedron and is sharp relative to the full joint distribution of
outcomes and regressors; no higher feature of the conditional outcome
distribution tightens the bounds on the coefficient mean and covariance.
When the support condition holds,
weighted minimum-distance estimators are root-$N$ asymptotically normal
and, under conditional nondegeneracy, the oracle generalized-inverse
weight attains the Chamberlain efficiency bound for the maintained
conditional-moment model.  Thus
finite-order moment identification remains a finite-dimensional problem
even when recovery of the full coefficient distribution is substantially
more demanding.

\appendix
\section{Proofs}\label{app:proofs}

\begin{proof}[Proof of Lemma~\ref{lem:rep}]
For each $t$,
\[
Y_t^{a_t}
=(\bW_t'\bD)^{a_t}
=
\sum_{|\bk_t|=a_t}
\binom{a_t}{\bk_t}
\bW_t^{\bk_t}\bD^{\bk_t}
\]
by the multinomial theorem.  Multiplying over $t=1,\ldots,T$ gives
\[
\bY^{\ba}
=
\sum_{|\bk_1|=a_1}\cdots\sum_{|\bk_T|=a_T}
\left(\prod_{t=1}^T\binom{a_t}{\bk_t}\right)
\left(\prod_{t=1}^T\bW_t^{\bk_t}\right)
\bD^{\sum_t\bk_t}.
\]
Collecting terms with $\bb=\sum_t\bk_t$, and using
$|\bb|=\sum_ta_t=r$, gives exactly the coefficient in
\eqref{eq:Rentry}.  Stacking over $\ba$ yields \eqref{eq:lift}.
\end{proof}

\begin{proof}[Proof of Lemma~\ref{lem:rank}]
For nonemptiness, $K_{\ba\bb}$ is the set of $T\times q$ arrays of
nonnegative integers with row sums $(a_1,\ldots,a_T)$ and column sums
$(b_1,\ldots,b_q)$.  The row and column totals both equal $r$, so such an
array exists; for example, the northwest-corner construction recursively
allocates the minimum remaining row and column total.

For the rank claim, suppose $c'\bR_r(\bW)=0'$ for some
$c\in\mathbb{R}^{S_r}$.  Then, for every $d\in\mathbb{R}^q$, letting
$y=\bW d$,
\[
0
=
c'\bR_r(\bW)d^{[r]}
=
c'y^{[r]}
=
\sum_{|\ba|=r}c_{\ba}y^{\ba}.
\]
If $\rank(\bW)=T$, the map $d\mapsto\bW d$ is surjective onto
$\mathbb{R}^T$.  Hence the homogeneous polynomial
$\sum_{|\ba|=r}c_{\ba}y^{\ba}$ vanishes for every
$y\in\mathbb{R}^T$.  The distinct degree-$r$ products indexed by $\ba$
are linearly independent as functions on $\mathbb{R}^T$, so $c=0$.
Therefore the rows of $\bR_r(\bW)$ are linearly independent.
\end{proof}

\begin{proof}[Proof of Lemma~\ref{lem:polar}]
Fix $\bx=w'\blam$ and let $d\in\mathbb R^q$ be an auxiliary vector.
Expanding $(\bx'd)^r$ in two ways gives
\[
(\bx'd)^r
=
\sum_{|\bb|=r}\binom r\bb\bx^{\bb}d^{\bb}
\]
and
\[
(\bx'd)^r
=
\big(\blam'(wd)\big)^r
=
\sum_{|\ba|=r}\binom r\ba\blam^{\ba}(wd)^{\ba}.
\]
Apply to both expressions the linear functional $L_{\bc}$ on
degree-$r$ polynomials in $d$ defined by
$L_{\bc}(d^{\bb})=c_{\bb}$.  The first expansion gives
$P_{\bc}(\bx)$.  Using
$(wd)^{\ba}=\sum_{\bb}\bR_{r,\ba\bb}(w)d^{\bb}$, the second gives the
right-hand side of \eqref{eq:polar}.  Since the distinct degree-$r$
products $\blam^{\ba}$ are linearly independent, the resulting polynomial
in $\blam$ vanishes identically if and only if every component of
$\bR_r(w)\bc$ is zero.
\end{proof}

\begin{proof}[Proof of Theorem~\ref{thm:support}]
If a nonzero $\bc$ satisfies $\bR_r(w)\bc=\bzero$ for every
$w\in\cW$, then for every admissible weight
\[
\bc'\bM_{r\bA}\bc
=
\E\big[(\bR_r\bc)'\bA_r(\bR_r\bc)\big]
=0,
\]
so $\bM_{r\bA}$ is singular.

Conversely, suppose condition (iii) holds.  For $\bc\neq0$, the map
$w\mapsto\bR_r(w)\bc$ is continuous.  Its zero set
$\{w:\bR_r(w)\bc=\bzero\}$ is closed and, by condition (iii), does not
contain $\cW$.  Its complement therefore contains an open neighborhood of
some point in $\cW$, which has positive probability by the definition of
support.  Hence
\[
\Pr\{\bR_r(\bW)\bc\neq\bzero\}>0.
\]
Since $\bR_r(\bW)\bc$ lies in
$\operatorname{col}(\bR_r(\bW))$, on which
$\bA_r(\bW)$ is positive definite, the integrand
$(\bR_r\bc)'\bA_r(\bR_r\bc)$ is nonnegative almost surely and strictly
positive on an event of positive probability.  Thus
$\bc'\bM_{r\bA}\bc>0$ for every admissible weight.  This proves the
equivalence of (i)--(iii).

By Lemma~\ref{lem:polar},
$\bR_r(w)\bc=\bzero$ for every $w\in\cW$ if and only if
$P_{\bc}$ vanishes on every $\rowsp(w)$, equivalently on $\cX$.
Since $\bc\mapsto P_{\bc}$ is a linear bijection, (iii) and (iv) are
equivalent, and the stated isomorphism between the unidentified moment
spaces follows.
\end{proof}

\begin{proof}[Proof of Theorem~\ref{thm:hyper}]
Write the distinct hyperplanes as $H_j=\{\ell_j=0\}$.
If there are $m\le r$, the nonzero homogeneous polynomial
\[
P=\ell_1\cdots\ell_m\ell_1^{r-m}
\]
has degree $r$ and vanishes on their union, so the support condition fails.

Conversely, a polynomial that vanishes identically on the hyperplane
$\{\ell=0\}$ is divisible by $\ell$.  To see this, choose coordinates in
which $\ell=x_q$.  If
$P(x_1,\ldots,x_{q-1},0)=0$ for every
$(x_1,\ldots,x_{q-1})$, every term of $P$ contains the factor $x_q$.
Dividing by one linear form preserves vanishing on the remaining
hyperplanes: each meets the complement of the divisor's zero set in a dense
subset, and the quotient is continuous, so it vanishes there as well.
Iterating, a degree-$r$ polynomial vanishing on $r+1$ distinct hyperplanes would
be divisible by the product of $r+1$ pairwise non-associate linear forms,
which is impossible unless the polynomial is zero.
Theorem~\ref{thm:support} completes the argument.
\end{proof}

\begin{proof}[Proof of Proposition~\ref{prop:twopoint}]
By Theorem~\ref{thm:support} applied to the two-point support
$\{w^{(1)},w^{(2)}\}$, the nullity of the vertically stacked matrix
\[
\begin{pmatrix}
\bR_r(w^{(1)})\\
\bR_r(w^{(2)})
\end{pmatrix}
\]
equals the dimension of the space of degree-$r$ homogeneous polynomials
vanishing on $L_1\cup L_2$.

Choose a basis $v_1,\ldots,v_q$ of $\mathbb R^q$ such that
\[
L_1\cap L_2=\operatorname{span}(v_1,\ldots,v_t),
\]
\[
L_1=\operatorname{span}(v_1,\ldots,v_T),
\]
and
\[
L_2
=
\operatorname{span}
(v_1,\ldots,v_t,v_{T+1},\ldots,v_{2T-t}).
\]
Such a basis exists because
$\dim(L_1+L_2)=2T-t\le q$.

In the associated coordinates, a homogeneous polynomial vanishes on a
coordinate subspace $L$ if and only if every one of its terms contains at
least one variable outside the coordinates spanning $L$.  Hence a basis for
the polynomials vanishing on both $L_1$ and $L_2$ consists of the
degree-$r$ products supported neither entirely on the coordinates of
$L_1$ nor entirely on those of $L_2$.  Inclusion--exclusion therefore gives
\[
\text{nullity}
=
Q_r-\binom{T+r-1}{r}-\binom{T+r-1}{r}+\binom{t+r-1}{r},
\]
because a degree-$r$ product supported in both coordinate sets is supported
on the coordinates of $L_1\cap L_2$.  Subtracting the nullity from $Q_r$
gives \eqref{eq:twopoint}.
\end{proof}

\begin{proof}[Proof of Proposition~\ref{prop:intercept}]
Because $\bar\bx$ belongs to the affine hull, there are coefficients
$\alpha_t$ satisfying
\[
\sum_t\alpha_t=1,
\qquad
\bar\bx=\sum_t\alpha_t\bX_t.
\]
Any $\bu\in U$ can be written as
$\bu=\sum_t\beta_t\bX_t$ for coefficients satisfying
$\sum_t\beta_t=0$.  Hence, for any $s\in\mathbb R$, choosing
$\lambda_t=s\alpha_t+\beta_t$ gives
\[
\sum_t\lambda_t(1,\bX_t')'
=
(s,s\bar\bx+\bu).
\]
This proves one inclusion of \eqref{eq:row-intercept}.

Conversely, for any $\blam\in\mathbb R^T$ set $s=\sum_t\lambda_t$; then $\sum_t\lambda_t(1,\bX_t')'=(s,\ s\bar\bx+\sum_t\lambda_t(\bX_t-\bar\bx))$, and $\bX_t-\bar\bx\in U(\omega)$ for every $t$ because $\bar\bx$ lies in the affine hull, so the second component lies in $s\bar\bx+U(\omega)$.  This proves \eqref{eq:row-intercept}.

If $s\neq0$, homogeneity gives
\[
P(s,s\bar\bx+\bu)
=
s^r p(\bar\bx+\bu/s).
\]
Since $U$ is a linear space, vanishing for every $s\neq0$ and
$\bu\in U$ is equivalent to $p$ vanishing on
$\bar\bx+U$, the affine hull.

At $s=0$,
\[
P(0,\bu)=p_r(\bu).
\]
This condition is implied by vanishing of $p$ on the affine hull.
For any $\bu\in U$, $p(\bar\bx+t\bu)$ is a polynomial in $t$ of degree at
most $r$, with coefficient on $t^r$ equal to $p_r(\bu)$.  Since
$\bar\bx+t\bu$ belongs to the affine hull for every $t$, this polynomial
vanishes identically, and therefore $p_r(\bu)=0$.
Applying the argument to every support realization proves the result.
\end{proof}

\begin{proof}[Proof of Proposition~\ref{prop:toporder}]
Suppose a nonzero homogeneous polynomial $P_s$ of degree $s\le r$ vanishes
on $\cX$.  For any nonzero linear form $\ell$, the product
$\ell^{r-s}P_s$ is a nonzero homogeneous polynomial of degree $r$ that also
vanishes on $\cX$.  Hence the support condition at order $r$ implies the
support condition at every lower order.

For the stacked system, positive definiteness of an admissible stacked
weight on the column space of $\bR_{\le r}(\bW)$ implies, by the argument in
Theorem~\ref{thm:support}, that the Gram matrix is singular if
and only if there exists a nonzero stacked vector
$\bc=(\bc_1',\ldots,\bc_r')'$ such that
\[
\bR_s(w)\bc_s=\bzero
\]
for every $s=1,\ldots,r$ and every $w\in\cW$.
Such a stacked null direction exists if and only if the support condition
fails at some order $s\le r$, which by the first part is equivalent to
failure at order $r$.
\end{proof}

\begin{proof}[Proof of Theorem~\ref{thm:ident}]
By Lemma~\ref{lem:rep},
\[
\bY^{[r]}=\bR_r\bD^{[r]}.
\]
Since $\bR_r'\bA_r\bR_r$ is $\bW$-measurable,
$\bD^{[r]}$ is integrable by Assumption~\ref{ass:moment}, and
$\bR_r'\bA_r\bR_r\bD^{[r]}
=\bR_r'\bA_r\bY^{[r]}$ is integrable by admissibility,
\[
\E[\bR_r'\bA_r\bR_r\bD^{[r]}\mid\bW]
=
\bR_r'\bA_r\bR_r\,\E[\bD^{[r]}\mid\bW].
\]
Therefore, by iterated expectations and Assumption~\ref{ass:moment},
\begin{align*}
\E[\bR_r'\bA_r\bY^{[r]}]
&=
\E[\bR_r'\bA_r\bR_r\bD^{[r]}]\\
&=
\E\left[
\bR_r'\bA_r\bR_r\E[\bD^{[r]}\mid\bW]
\right]\\
&=
\E[\bR_r'\bA_r\bR_r]\bdel^{[r]}
=
\bM_{r\bA}\bdel^{[r]}.
\end{align*}
Theorem~\ref{thm:support} implies that $\bM_{r\bA}$ is nonsingular for
every admissible weight, which gives \eqref{eq:ident}.
\end{proof}

\begin{proof}[Proof of Theorem~\ref{thm:asynorm}]
Let
\[
e_{r,i}=\bY_i^{[r]}-\bR_{r,i}\bdel^{[r]}.
\]
The estimator satisfies the exact expansion
\[
\sqrt N(\widehat{\bdel}^{[r]}-\bdel^{[r]})
=
\widehat\bM_{r\bA}^{-1}
\frac1{\sqrt N}\sum_{i=1}^N
\bR_{r,i}'\bA_{r,i}e_{r,i}.
\]
By Assumption~\ref{ass:moment},
\[
\E[e_{r,i}\mid\bW_i]
=
\bR_{r,i}
\left(
\E[\bD_i^{[r]}\mid\bW_i]-\bdel^{[r]}
\right)
=0,
\]
and hence $\E[\bpsi_{r,i}]=0$.

Assumption~\ref{ass:sampling} and \eqref{eq:int1} imply a WLLN for
$\widehat\bM_{r\bA}$.  The numerator also satisfies a WLLN because
\[
\bR_r'\bA_r\bY^{[r]}
=
\bR_r'\bA_r\bR_r\bdel^{[r]}+\bpsi_r
\]
has a finite first moment under \eqref{eq:int1}--\eqref{eq:int2}.
Condition~\eqref{eq:int2} gives a finite second moment for the moment
contribution $\bpsi_r$, so the multivariate CLT applies to
$N^{-1/2}\sum_i\bpsi_{r,i}$.
Since
$\widehat\bM_{r\bA}\pto\bM_{r\bA}$ and $\bM_{r\bA}$ is nonsingular,
Slutsky's theorem gives \eqref{eq:clt}.

We next verify that $\bGam_r(\bW)$ in \eqref{eq:Gamma} is finite almost
surely under \eqref{eq:int2}.
Since $\bY^{[r]}=\bR_r\bD^{[r]}$, the residual
$e_r=\bR_r(\bD^{[r]}-\bdel^{[r]})$ lies in $\operatorname{col}(\bR_r)$.
The map $e\mapsto\bR_r'\bA_re$ is injective on
$\operatorname{col}(\bR_r)$: if $e=\bR_rc\neq\bzero$ and
$\bR_r'\bA_re=\bzero$, then $0=c'\bR_r'\bA_re=e'\bA_re$, contradicting
positive definiteness of $\bA_r$ on $\operatorname{col}(\bR_r)$.
For almost every fixed $\bW=w$, finite dimensionality therefore yields a
finite constant $C(w)$ such that $\|e\|\le C(w)\|\bR_r'\bA_re\|$ for all
$e\in\operatorname{col}(\bR_r(w))$.
Because \eqref{eq:int2} states $\E\|\bpsi_r\|^2<\infty$, the conditional
expectation $\E[\|\bpsi_r\|^2\mid\bW]$ is finite almost surely, so
\[
\E[\|e_r\|^2\mid\bW]
\le
C(\bW)^2\,\E[\|\bpsi_r\|^2\mid\bW]
<\infty
\quad\text{almost surely,}
\]
and $\bGam_r(\bW)=\Var(e_r\mid\bW)$ is finite almost surely.

Finally,
\[
\E[\bpsi_r\bpsi_r'\mid\bW]
=
\bR_r'\bA_r\Var(\bY^{[r]}\mid\bW)\bA_r\bR_r,
\]
which yields \eqref{eq:V} by iterated expectations.
Consistency follows from the same WLLN and continuous-mapping argument
without the CLT.  That argument requires only \eqref{eq:int1} and
$\E\|\bpsi_r\|<\infty$, proving the final claim.
\end{proof}

\begin{proof}[Proof of Theorem~\ref{thm:robust}]
Write
\[
\widehat e_{r,i}
=
e_{r,i}
-
\bR_{r,i}(\widehat\bdel^{[r]}-\bdel^{[r]}).
\]
Therefore
\[
\widehat\bpsi_{r,i}
=
\bpsi_{r,i}
-
\bR_{r,i}'\bA_{r,i}\bR_{r,i}
(\widehat\bdel^{[r]}-\bdel^{[r]}).
\]
Let
\[
\bH_i=\bR_{r,i}'\bA_{r,i}\bR_{r,i},
\qquad
\Delta=\widehat\bdel^{[r]}-\bdel^{[r]}=O_p(N^{-1/2}).
\]
Then
\[
\frac1N\sum_i\widehat\bpsi_{r,i}\widehat\bpsi_{r,i}'
=
\frac1N\sum_i\bpsi_{r,i}\bpsi_{r,i}'
-\frac1N\sum_i\bpsi_{r,i}\Delta'\bH_i
-\frac1N\sum_i\bH_i\Delta\bpsi_{r,i}'
+\frac1N\sum_i\bH_i\Delta\Delta'\bH_i.
\]
The leading term converges to $\bV_{r\bA}$ by the WLLN under
\eqref{eq:int2}.  For the cross terms, Cauchy--Schwarz gives
\[
\Big\|\frac1N\sum_i\bpsi_{r,i}\Delta'\bH_i\Big\|
\le
\|\Delta\|
\left(\frac1N\sum_i\|\bpsi_{r,i}\|^2\right)^{1/2}
\left(\frac1N\sum_i\|\bH_i\|^2\right)^{1/2}.
\]
The two sample averages are $O_p(1)$ by \eqref{eq:int2} and
\eqref{eq:int3}, respectively, while $\|\Delta\|=o_p(1)$.
Thus the cross terms are $o_p(1)$.  The final term satisfies
\[
\left\|
\frac1N\sum_i\bH_i\Delta\Delta'\bH_i
\right\|
\le
\|\Delta\|^2
\frac1N\sum_i\|\bH_i\|^2
=o_p(1)
\]
by \eqref{eq:int3}.  Hence \eqref{eq:Vhatcons}.
The sandwich convergence follows by continuous mapping and nonsingularity
of $\bM_{r\bA}$.
\end{proof}

\begin{proof}[Proof of Theorem~\ref{thm:efficient}]
Let
\[
e
=
\bY^{[r]}-\bR_r\bdel^{[r]}
=
\bR_r(\bD^{[r]}-\bdel^{[r]})
\]
and define
\[
u=\bR_r'\bA_r e,
\qquad
v=\bR_r'\bGam_r^{+}e.
\]
Under Assumption~\ref{ass:nondeg},
\[
\operatorname{col}(\bGam_r)
=
\operatorname{col}(\bR_r)
\quad\text{a.s.},
\]
so
$\bGam_r\bGam_r^{+}\bR_r=\bR_r$ almost surely.
Together with
$\bGam_r^{+}\bGam_r\bGam_r^{+}=\bGam_r^{+}$ and
$\E[ee'\mid\bW]=\bGam_r$, this gives
\begin{align*}
\E[uv']
&=
\E[\bR_r'\bA_r\bGam_r\bGam_r^{+}\bR_r]
=
\bM_{r\bA},\\
\E[vv']
&=
\E[\bR_r'\bGam_r^{+}\bGam_r\bGam_r^{+}\bR_r]
=
\E[\bR_r'\bGam_r^{+}\bR_r]
=
\bB_r.
\end{align*}
Both expectations are finite:
$\E\|u\|^2<\infty$ by \eqref{eq:int2}, while
$\E\|v\|^2=\operatorname{tr}\bB_r<\infty$ by
Assumption~\ref{ass:nondeg}.

Positive semidefiniteness of
\[
\E\left[
(u-\bM_{r\bA}\bB_r^{-1}v)
(u-\bM_{r\bA}\bB_r^{-1}v)'
\right]
\]
implies
\[
\bV_{r\bA}
=
\E[uu']
\succeq
\bM_{r\bA}\bB_r^{-1}\bM_{r\bA}.
\]
Pre- and post-multiplying by $\bM_{r\bA}^{-1}$ gives
\eqref{eq:effineq}.

If $\bA_r=\bGam_r^{+}$, then $u=v$, so
\[
\bV_{r\bA^*}
=
\bM_{r\bA^*}
=
\bB_r.
\]
The asymptotic covariance is therefore $\bB_r^{-1}$.
The weight $\bGam_r^{+}$ is admissible under the conditions stated in
Theorem~\ref{thm:efficient}, so Theorem~\ref{thm:asynorm} applies.

It remains to establish the Chamberlain bound.
For the conditional moment restriction \eqref{eq:cmr}, the derivative of
the conditional mean with respect to $\bdel^{[r]}$ is
$-\bR_r(\bW)$ and the conditional covariance is $\bGam_r(\bW)$.
For $F_W$-almost every $w$ with $\rank(w)=T$,
$\bGam_r(w)$ is nonsingular, and Chamberlain's (1987) optimal-instrument
formula gives conditional information
\[
\bR_r(w)'\bGam_r(w)^{-1}\bR_r(w)
=
\bB_r(w).
\]

Fix $w\in\cW$ with $\rank(w)=k<T$ at which $\bSig_r(w)$ is positive definite; Assumption~\ref{ass:nondeg} guarantees this for $F_W$-almost every $w$.
Let $w_\circ$ consist of $k$ linearly independent rows of $w$.
Since every row of $w$ is a linear combination of the rows of $w_\circ$,
there is a $T\times k$ matrix $\mathbf F$ of rank $k$ such that
\[
w=\mathbf Fw_\circ.
\]
If
$\bY_\circ=w_\circ\bD$, then
$\bY=\mathbf F\bY_\circ$ and Lemma~\ref{lem:rep} gives
\[
\bY^{[r]}
=
\bR_r(\mathbf F)\bY_\circ^{[r]},
\qquad
\bR_r(w)
=
\mathbf E\,\bR_r(w_\circ),
\]
where
\[
\mathbf E=\bR_r(\mathbf F)
\]
has dimension
$S_r\times\binom{k+r-1}{r}$.

The matrix $\mathbf E$ has full column rank.
Indeed, by Lemma~\ref{lem:polar},
$\bR_r(\mathbf F)\bc=\bzero$ if and only if
$P_{\bc}$ vanishes on
$\rowsp(\mathbf F)=\mathbb R^k$, which implies $\bc=\bzero$.
The conditional restriction \eqref{eq:cmr} at $w$ is therefore equivalent
to
\[
\E[
\bY_\circ^{[r]}
-
\bR_r(w_\circ)\bdel^{[r]}
\mid\bW=w
]
=
\bzero.
\]
Its conditional covariance is
\[
\bGam_\circ
=
\bR_r(w_\circ)\bSig_r(w)\bR_r(w_\circ)',
\]
which is nonsingular by Lemma~\ref{lem:rank} and positive definiteness of
$\bSig_r(w)$.

Chamberlain's formula for this reduced set of linearly independent
restrictions gives conditional information
\[
\bR_r(w_\circ)'\bGam_\circ^{-1}\bR_r(w_\circ).
\]
Since
\[
\bGam_r(w)
=
\mathbf E\bGam_\circ\mathbf E'
\]
with $\mathbf E$ of full column rank and $\bGam_\circ$ nonsingular,
\[
\bGam_r(w)^{+}
=
(\mathbf E')^{+}
\bGam_\circ^{-1}
\mathbf E^{+},
\qquad
\mathbf E^{+}
=
(\mathbf E'\mathbf E)^{-1}\mathbf E',
\]
and $\mathbf E^{+}\mathbf E=\bI$.  Hence
\[
\bR_r(w)'\bGam_r(w)^{+}\bR_r(w)
=
\bR_r(w_\circ)'\mathbf{E}'(\mathbf{E}')^{+}\bGam_\circ^{-1}\mathbf{E}^{+}\mathbf{E}\bR_r(w_\circ)
=
\bR_r(w_\circ)'\bGam_\circ^{-1}\bR_r(w_\circ).
\]
Thus $\bB_r(w)$ in \eqref{eq:Bcond} is the Chamberlain conditional
information matrix for $F_W$-almost every $w$.
Averaging over $\bW$ gives the information matrix $\bB_r$ and
the efficiency bound $\bB_r^{-1}$.
\end{proof}

\begin{proof}[Proof of Corollary~\ref{cor:infogeom}]
(i)
The matrix $\bB_r$ is finite by assumption and positive semidefinite.
Thus $\bc\in\ker\bB_r$ if and only if
$\bc'\bB_r\bc=0$.  For $\bc\in\mathbb R^{Q_r}$,
\[
\bc'\bB_r\bc
=
\E\big[(\bR_r\bc)'\bGam_r^{+}(\bR_r\bc)\big].
\]
The integrand is nonnegative.
Moreover, $\bGam_r(\bW)^{+}$ is positive definite on
$\operatorname{col}(\bR_r(\bW))$ almost surely because
$\operatorname{col}(\bGam_r)=\operatorname{col}(\bR_r)$ under
Assumption~\ref{ass:nondeg}.
Hence
$\bc'\bB_r\bc=0$ if and only if
$\bR_r(\bW)\bc=\bzero$ almost surely.

The map $w\mapsto\bR_r(w)\bc$ is polynomial in the entries of $w$ and
therefore continuous.  Its zero set is closed.  If that set has probability
one, it contains $\cW=\suppp(\bW)$.
Conversely, if $\bR_r(w)\bc=\bzero$ for every $w\in\cW$, then
$\bR_r(\bW)\bc=\bzero$ almost surely.
Therefore
\[
\ker\bB_r
=
\{\bc:\bR_r(w)\bc=\bzero\ \text{for every }w\in\cW\}
=
\bigcap_{w\in\cW}\ker\bR_r(w)
=
\cN_r,
\]
by \eqref{eq:Nr}.
Lemma~\ref{lem:polar} identifies $\cN_r$ with $\cI_r(\cX)$ through
$\bc\mapsto P_{\bc}$.
Theorem~\ref{thm:support} then implies that $\bB_r$ is nonsingular if and
only if the support condition holds.

(ii)
Fix $w$ such that $\bSig_r(w)$ is positive definite.
Then
\[
\operatorname{col}(\bGam_r(w))
=
\operatorname{col}(\bR_r(w)).
\]
If $\bB_r(w)\bc=\bzero$, then
\[
\bc'\bB_r(w)\bc
=
(\bR_r(w)\bc)'\bGam_r(w)^{+}(\bR_r(w)\bc)
=
0.
\]
Because $\bR_r(w)\bc$ lies in
$\operatorname{col}(\bGam_r(w))$, on which
$\bGam_r(w)^{+}$ is positive definite, this implies
$\bR_r(w)\bc=\bzero$.
The converse is immediate.
Hence
\[
\ker\bB_r(w)=\ker\bR_r(w)
\]
and
\[
\rank\bB_r(w)
=
\rank\bR_r(w)
=
\binom{\rank(w)+r-1}{r}
\]
by \eqref{eq:general-rank}.
For $\rank(w)=1$, this rank equals one.
If $T<q$, then $\rank(w)\le T$, and therefore
\[
\rank\bB_r(w)
\le
\binom{T+r-1}{r}
=
S_r
<
Q_r.
\]
Thus $\bB_r(w)$ is singular at every such $w$.
Nonsingularity of the matrix $\bB_r$ under the support condition
follows from part (i).
\end{proof}

\begin{proof}[Proof of Corollary~\ref{cor:exactrecover}]
If $q=T$ and $\rank(\bW)=T$, then $S_r=Q_r$ and
Lemma~\ref{lem:rank} implies that $\bR_r$ is invertible.
Equation~\eqref{eq:Drecover} therefore follows from \eqref{eq:lift}.

Under conditional homoskedasticity,
\[
\bGam_r
=
\bR_r\bSig_r\bR_r',
\qquad
\bGam_r^{-1}
=
\bR_r'^{-1}\bSig_r^{-1}\bR_r^{-1},
\]
so
\[
\bR_r'\bGam_r^{-1}\bR_r
=
\bSig_r^{-1}
\]
is constant across observations.
The weighted estimator therefore reduces to the sample mean of the
recovered $\bD_i^{[r]}$, and the stated limiting distribution follows from
the multivariate CLT.

For \eqref{eq:qTharmonic},
\[
\bR_r'\bGam_r^{-1}\bR_r
=
\bSig_r(\bW)^{-1}
\]
holds realization by realization without conditional homoskedasticity.
Thus
\[
\bM_{r\bA^*}
=
\bV_{r\bA^*}
=
\E[\bSig_r(\bW)^{-1}],
\]
and the sandwich covariance reduces to
\[
\big(\E[\bSig_r(\bW)^{-1}]\big)^{-1}.
\]
Since
$\E[\bD^{[r]}\mid\bW]=\bdel^{[r]}$ by
Assumption~\ref{ass:moment}, the variance decomposition gives
\[
\Var(\bD^{[r]})
=
\E[\bSig_r(\bW)],
\]
which is the asymptotic covariance of
$\overline{\bD^{[r]}}$.

For the matrix inequality, let
\[
\bB=(\E[\bSig_r])^{-1}.
\]
Almost surely,
\[
\big(\bSig_r^{-1/2}-\bSig_r^{1/2}\bB\big)'
\big(\bSig_r^{-1/2}-\bSig_r^{1/2}\bB\big)
=
\bSig_r^{-1}-2\bB+\bB\bSig_r\bB
\succeq
\bzero.
\]
Taking expectations gives
\[
\E[\bSig_r^{-1}]
\succeq
2\bB-\bB\E[\bSig_r]\bB
=
(\E[\bSig_r])^{-1}.
\]
Inverting the ordering of positive definite matrices yields
\[
(\E[\bSig_r^{-1}])^{-1}
\preceq
\E[\bSig_r].
\]
Equality implies that the expectation of the positive semidefinite matrix
above is zero and therefore
\[
\bSig_r^{-1/2}
=
\bSig_r^{1/2}\bB
\quad\text{a.s.}
\]
Thus
$\bSig_r(\bW)=\E[\bSig_r]$ almost surely.
The converse is immediate.
\end{proof}

\begin{proof}[Proof of Lemma~\ref{lem:dims}]
The number of distinct products of total degree at most $r$ in $T$
variables is $\binom{T+r}{r}$.
Appending a slack variable that absorbs the deficient degree gives a
bijection with the degree-$r$ products in $T+1$ variables, whose number is
\[
\binom{(T+1)+r-1}{r}.
\]
Removing the degree-zero product gives
\[
S_{\le r}=\binom{T+r}{r}-1.
\]
The same argument with $q$ gives $Q_{\le r}$.
Since $\bR_{\le r}$ is block diagonal and each block $\bR_s$ has row rank
$S_s$ whenever $\rank(\bW)=T$, its rank is
\[
\sum_{s=1}^rS_s=S_{\le r}
\]
almost surely.
\end{proof}

\begin{proof}[Proof of Theorem~\ref{thm:stacked}]
Apply Theorems~\ref{thm:ident}--\ref{thm:robust} to the stacked equation
\eqref{eq:stackrep}.
Assumption~\ref{ass:momentstack} gives
\[
\E[\bD^{[\le r]}\mid\bW]=\bdel^{[\le r]},
\]
while Proposition~\ref{prop:toporder} and
Theorem~\ref{thm:support} give nonsingularity of
$\bM_{\le r,\bA}$.
The covariance of the stacked moment contributions and the oracle bound
follow from the same calculations as in
Theorems~\ref{thm:asynorm} and \ref{thm:efficient}; in particular, the
argument establishing almost-sure finiteness of $\bGam_r(\bW)$ in the
proof of Theorem~\ref{thm:asynorm} applies verbatim with
$\bR_{\le r}$ and $\bA_{\le r}$ in place of $\bR_r$ and $\bA_r$, so
$\bGam_{\le r}(\bW)$ is finite almost surely under the stacked analogue
of \eqref{eq:int2}.
In particular, positive definiteness of
$\bSig_{\le r}(\bW)$ implies
\[
\operatorname{col}(\bGam_{\le r})
=
\operatorname{col}(\bR_{\le r}),
\]
so
\[
\bGam_{\le r}\bGam_{\le r}^{+}\bR_{\le r}
=
\bR_{\le r}.
\]
The reduction to linearly independent rows used in the proof of
Theorem~\ref{thm:efficient} then applies to each block of
$\bR_{\le r}$.
\end{proof}

\begin{proof}[Proof of Lemma~\ref{lem:conversion}]
The multivariate binomial theorem gives
\[
(\bD-\bmu)^{\ba}
=
\sum_{\bb\le\ba}
\binom{\ba}{\bb}
\bD^{\bb}(-\bmu)^{\ba-\bb}.
\]
Taking expectations yields \eqref{eq:rawcentral}.
The inverse follows by applying the same expansion to
$\bD=\bDt+\bmu$ and using
$\E[\bDt^{e_j}]=0$.
\end{proof}

\begin{proof}[Proof of Lemma~\ref{lem:collected}]
In \eqref{eq:rawcentral}, separate the terms with
$|\bb|=0$, $|\bb|=1$, and $|\bb|>1$.
Because
$\delta_{\bzero}=1$ and $\delta_{e_j}=\mu_j$, the first two groups are
proportional to $\bmu^{\ba}$.
Their combined coefficient is
\[
(-1)^{|\ba|}
+
\sum_{j=1}^q(-1)^{|\ba|-1}a_j
=
(-1)^{|\ba|}(1-|\ba|),
\]
which gives \eqref{eq:collected}.
\end{proof}

\begin{proof}[Proof of Lemma~\ref{lem:derivs}]
Differentiate \eqref{eq:collected}, treating $\bmu$ and the raw moments
$\delta_{\bb}$ with $|\bb|>1$ as free coordinates.
The derivative of $\bmu^{\ba-\bb}$ with respect to $\mu_k$ is
\[
(a_k-b_k)\bmu^{\ba-\bb-e_k}
\]
when $b_k\le a_k-1$, and zero otherwise.
This gives \eqref{eq:dgdm}.
Equation~\eqref{eq:dgdd} is the coefficient on $\delta_{\bb}$ in
\eqref{eq:collected}.
\end{proof}

\begin{proof}[Proof of Lemma~\ref{lem:jacobian}]
A central moment of total degree $s$ depends only on raw moments of degree
at most $s$.
Hence $\bG$ is block lower triangular when the coordinates are ordered by
degree.
Within degree $s$, the coefficient on $\delta_{\ba}$ in
$\gamma_{\ba}$ is one, while no other degree-$s$ raw moment appears because
$\bb\le\ba$ and $|\bb|=|\ba|$ imply $\bb=\ba$.
Thus every diagonal block is an identity matrix and
$\det(\bG)=1$.
The inverse polynomial map is given explicitly by
\eqref{eq:centralraw}.
\end{proof}

\begin{proof}[Proof of Theorem~\ref{thm:delta}]
The map from $\bdel^{[\le r]}$ to $\bth^{[\le r]}$ is polynomial and hence
continuously differentiable, with Jacobian $\bG$.
Applying the multivariate delta method to Theorem~\ref{thm:stacked} gives
\eqref{eq:deltaclt}.
Consistency of $\widehat\bG$ follows from continuity of $\bG$ and
consistency of $\widehat\bdel^{[\le r]}$.
Combining this result with consistency of the stacked sandwich covariance
estimator in Theorem~\ref{thm:stacked}, which holds under the stacked
analogue of \eqref{eq:int3}, gives the stated plug-in covariance
estimator.
\end{proof}

\begin{proof}[Proof of Theorem~\ref{thm:partialset}]
($\subseteq$)
Any structure satisfying the maintained model with moment vector $\bdel^{[\le r]}$ that
reproduces \eqref{eq:info} satisfies, for every $s$,
\[
\bR_s(\bW)\big(\bdel^{[s]}-\bdel_0^{[s]}\big)=\bzero
\quad\text{a.s.}
\]
The map
$w\mapsto\bR_s(w)(\bdel^{[s]}-\bdel_0^{[s]})$ is continuous, so its zero
set is closed.  Since that set has probability one, it contains
$\cW=\suppp(\bW)$.  Hence
$\bdel^{[s]}-\bdel_0^{[s]}\in\cN_s$.
Moreover, $\bdel^{[\le r]}$ is the moment vector of the coefficient
distribution under that structure, so
$\bdel^{[\le r]}\in\cM_{\le r}$.

($\supseteq$)
Let $\bdel^{[\le r]}\in\Theta_I$.
By the definition of $\cM_{\le r}$, there exists a distribution $F$ with
moments $\bdel^{[\le r]}$.
Let $\bD\sim F$ independently of $\bW$.
Then Assumption~\ref{ass:momentstack} holds and, for $F_W$-almost every
$w$,
\[
\E[\bY^{[s]}\mid\bW=w]
=
\bR_s(w)\bdel^{[s]}
=
\bR_s(w)\bdel_0^{[s]}
=
m_s(w),
\]
because
$\bdel^{[s]}-\bdel_0^{[s]}\in\cN_s$.
This proves both inclusion and the independence-attainability claim.

Convexity follows because $\cM_{\le r}$ is convex: the moment map is
linear in the underlying distribution and mixtures of probability
distributions remain probability distributions.
The affine set
$\bdel_0^{[\le r]}+\cN_{\le r}$ is also convex.

For the functional claim, if
$\ba\perp\cN_{\le r}$, then
$\ba'\bdel^{[\le r]}$ is constant on every affine set
$\bdel_0^{[\le r]}+\cN_{\le r}$ and therefore on the corresponding
identified set.

Conversely, suppose
$\ba'\mathbf n\neq0$ for some
$\mathbf n\in\cN_{\le r}$.
The stacked moment vectors of point masses
\[
\left\{
(\bx^{[1]\prime},\ldots,\bx^{[r]\prime})':
\bx\in\mathbb R^q
\right\}
\]
affinely span $\mathbb R^{Q_{\le r}}$.
Otherwise, a nontrivial affine relation among these vectors would give a
nonzero polynomial of degree at most $r$ that vanishes on all of
$\mathbb R^q$.
Consequently, one can choose $Q_{\le r}+1$ points whose stacked moment
vectors are affinely independent.
A strictly positive convex combination of those vectors lies in the
interior of their full-dimensional simplex, which is contained in
$\cM_{\le r}$.

Choose such a vector as $\bdel_0^{[\le r]}$ and let $\bD$ have the
corresponding finite-support distribution independently of $\bW$.
For sufficiently small $\varepsilon>0$,
\[
\bdel_0^{[\le r]}\pm\varepsilon\mathbf n
\in
\cM_{\le r}.
\]
Because $\mathbf n\in\cN_{\le r}$, both vectors also belong to the same
affine observational-equivalence class and hence to $\Theta_I$, while
$\ba'\bdel^{[\le r]}$ differs between them.
Thus a linear functional is identified under the maintained
conditional-moment model if and only if its coefficient vector is
orthogonal to $\cN_{\le r}$.
The stacked dimension count is \eqref{eq:hilbertstack}.
\end{proof}

\begin{proof}[Proof of Theorem~\ref{thm:spectra}]
Necessity follows because, for any distribution $F$ of $\bD$ with the
candidate moments,
\[
M(\bdel^{[\le2]})
=
\E[(1,\bD')'(1,\bD')]
\succeq0.
\]
Conversely, if
$M(\bdel^{[\le2]})\succeq0$, the Schur complement gives
\[
\operatorname{mat}(\bdel^{[2]})
-
\bdel^{[1]}(\bdel^{[1]})'
\succeq0.
\]
A Gaussian distribution with mean $\bdel^{[1]}$ and this covariance matrix
has exactly the candidate first and second moments.
Hence
\[
\cM_{\le2}
=
\{\bdel^{[\le2]}:
M(\bdel^{[\le2]})\succeq0\},
\]
and Theorem~\ref{thm:partialset} gives the spectrahedron
\eqref{eq:spectra}.

The set is closed because
$\bdel^{[\le2]}\mapsto M(\bdel^{[\le2]})$ is affine and continuous and the
positive-semidefinite cone is closed.
Optimizing a linear functional over this affine set subject to the linear
matrix inequality is therefore a semidefinite program.
Theorem~\ref{thm:partialset} makes the resulting infimum and supremum sharp,
and the Gaussian construction shows that every feasible moment vector is
attained by a structure satisfying the maintained model.
If an optimum is attained, its endpoint is therefore attained by a possibly
singular Gaussian structure.
\end{proof}

\begin{proof}[Proof of Corollary~\ref{cor:sharpdata}]
Fix $w\in\cW$ with $\rank(w)=k$, write
\[
m(w)=\E[\bY\mid\bW=w],
\qquad
\boldsymbol{\Omega}(w)=\Var(\bY\mid\bW=w),
\]
and abbreviate
$\mathbf O=\mathbf O_w$ and $\mathbf N=\mathbf N_w$.
Since $\bY=w\bD$,
$\bY\in\operatorname{col}(w)=\operatorname{col}(w\mathbf O)$ almost surely.
Because $w\mathbf O$ has full column rank,
\[
(w\mathbf O)^{+}
=
\big((w\mathbf O)'(w\mathbf O)\big)^{-1}(w\mathbf O)'
\]
and
\[
w\mathbf O(w\mathbf O)^{+}\bY
=
\bY
\quad\text{a.s.}
\]

\emph{Restrictions implied by $\Theta_I$ at $w$.}
Write a candidate as $(\bmu,\bSig)$, where
\[
\bmu=\bdel^{[1]},
\qquad
\bSig
=
\operatorname{mat}(\bdel^{[2]})-\bmu\bmu'.
\]
Membership in
$\bdel_0^{[\le2]}+\cN_{\le2}$ implies
\[
\bR_1(w)(\bmu-\bmu_0)=\bzero
\]
and
\[
\bR_2(w)(\bdel^{[2]}-\bdel_0^{[2]})=\bzero.
\]
By \eqref{eq:vech}, these conditions are equivalent to
\[
w\bmu=m(w),
\qquad
w\bSig w'=\boldsymbol{\Omega}(w).
\]
Membership in $\cM_{\le2}$ is equivalent to $\bSig\succeq0$ by
Theorem~\ref{thm:spectra}.

Partition $\bSig$ in the coordinates $(\mathbf O,\mathbf N)$:
\[
\bSig_{OO}=\mathbf{O}'\bSig\mathbf{O},
\qquad
\bSig_{NO}=\mathbf{N}'\bSig\mathbf{O},
\qquad
\bSig_{NN}=\mathbf{N}'\bSig\mathbf{N}.
\]
Since
$w=w\mathbf O\mathbf O'$, the affine restrictions imply
\[
\mathbf{O}'\bmu=(w\mathbf{O})^{+}m(w),
\qquad
\bSig_{OO}
=
(w\mathbf{O})^{+}\boldsymbol{\Omega}(w)(w\mathbf{O})^{+\prime}.
\]
Thus the row-space components of the candidate mean and covariance are fixed
by the observed conditional first and second moments, while
$\mathbf N'\bmu$, $\bSig_{NO}$, and $\bSig_{NN}$ are restricted only by
$\bSig\succeq0$.
For the partitioned covariance matrix, this is equivalent to
\begin{equation}\label{eq:albert}
\bSig_{OO}\succeq0,
\qquad
\bSig_{NO}=\bSig_{NO}\bSig_{OO}^{+}\bSig_{OO},
\qquad
\bSig_{NN}-\bSig_{NO}\bSig_{OO}^{+}\bSig_{ON}\succeq0
\end{equation}
(Albert, 1969), where $\bSig_{ON}=\bSig_{NO}'$.

\emph{Construction.}
Enlarge the probability space to carry
$\boldsymbol{\eta}\sim N(\bzero,\bI_{q-k})$ independently of
$(\bY,\bW)$, and set
\[
\bD
=
\mathbf{O}(w\mathbf{O})^{+}\bY+\mathbf{N}\mathbf{U},
\qquad
\mathbf{U}
=
\mathbf{a}+\mathbf{B}\bY+\mathbf{C}\boldsymbol{\eta},
\]
where
\[
\mathbf{B}=\bSig_{NO}\bSig_{OO}^{+}(w\mathbf{O})^{+},
\qquad
\mathbf{C}\mathbf{C}'=\bSig_{NN}-\bSig_{NO}\bSig_{OO}^{+}\bSig_{ON},
\qquad
\mathbf{a}=\mathbf{N}'\bmu-\mathbf{B}m(w).
\]
The matrix $\mathbf C$ exists by the third condition in
\eqref{eq:albert}.
Then
\[
w\bD
=
w\mathbf{O}(w\mathbf{O})^{+}\bY+w\mathbf{N}\mathbf{U}
=
\bY
\quad\text{almost surely},
\]
because $w\mathbf N=\bzero$.
Thus the construction preserves the conditional distribution of
$\bY$ given $\bW=w$.

\emph{Moments.}
Since
\[
\mathbf O'\bD=(w\mathbf O)^+\bY,
\qquad
\mathbf N'\bD=\mathbf U,
\]
we have
\[
\mathbf{O}'\E[\bD\mid\bW=w]
=
(w\mathbf{O})^{+}m(w)
=
\mathbf{O}'\bmu,
\]
and
\[
\mathbf{N}'\E[\bD\mid\bW=w]
=
\mathbf a+\mathbf Bm(w)
=
\mathbf N'\bmu.
\]
Hence
$\E[\bD\mid\bW=w]=\bmu$.

For the conditional covariance,
\begin{align*}
\Var(\mathbf{O}'\bD\mid\bW=w)
&=(w\mathbf{O})^{+}\boldsymbol{\Omega}(w)(w\mathbf{O})^{+\prime}
=\bSig_{OO},\\
\Cov(\mathbf{N}'\bD,\mathbf{O}'\bD\mid\bW=w)
&=\mathbf{B}\boldsymbol{\Omega}(w)(w\mathbf{O})^{+\prime}
=\bSig_{NO}\bSig_{OO}^{+}\bSig_{OO}
=\bSig_{NO},\\
\Var(\mathbf{N}'\bD\mid\bW=w)
&=\mathbf{B}\boldsymbol{\Omega}(w)\mathbf{B}'
+\mathbf{C}\mathbf{C}'\\
&=
\bSig_{NO}\bSig_{OO}^{+}\bSig_{OO}\bSig_{OO}^{+}\bSig_{ON}
+\mathbf{C}\mathbf{C}'
=
\bSig_{NN},
\end{align*}
where the second equality uses the second condition in
\eqref{eq:albert}, and the third uses
$\bSig_{OO}^{+}\bSig_{OO}\bSig_{OO}^{+}=\bSig_{OO}^{+}$ together with
the same condition.
Therefore
\[
\Var(\bD\mid\bW=w)=\bSig,
\]
and hence
\[
\E[\bD^{[\le2]}\mid\bW=w]
=
\bdel^{[\le2]}.
\]

The pointwise construction can be chosen measurably in $w$.
Indeed, the rank of $w$ takes finitely many values; on each rank stratum,
orthonormal bases $\mathbf O_w$ and $\mathbf N_w$ may be chosen measurably,
for example through a measurable QR or singular-value decomposition.
The Moore--Penrose inverse and the principal positive-semidefinite square
root used to define $\mathbf C$ are measurable functions of their matrix
arguments.  Thus $\mathbf O_w$, $\mathbf N_w$, $\mathbf a(w)$,
$\mathbf B(w)$, and $\mathbf C(w)$ can be chosen measurably.

Carrying out the construction for $F_W$-almost every $w$ therefore gives a
conditional distribution satisfying Assumption~\ref{ass:momentstack} and
reproducing the observed conditional distribution of $\bY$.
The identified set relative to the full joint distribution of
$(\bY,\bW)$ is contained in $\Theta_I$ because the full distribution
contains the conditional-moment information \eqref{eq:info}, and it
contains $\Theta_I$ by the construction.
\end{proof}

\begin{proof}[Proof of Corollary~\ref{cor:compact}]
By Theorem~\ref{thm:spectra}, $\Theta_I$ is closed and convex.
Hence it is compact if and only if its recession cone is
$\{\bzero\}$ (Rockafellar, 1970, Theorem~8.4).
Since
\[
\Theta_I
=
(\bdel_0^{[\le2]}+\cN_{\le2})\cap\cM_{\le2}
\]
is nonempty and $\cN_{\le2}$ is a linear subspace, the recession cone of
the intersection is the intersection of the recession cones
(Rockafellar, 1970, Corollary~8.3.3).

\emph{Recession cone of $\cM_{\le2}$.}
The map in \eqref{eq:momentmatrix} is affine:
\[
M(\bdel+t\mathbf d)
=
M(\bdel)+tM_0(\mathbf d),
\]
where
\[
M_0(\mathbf{d})
=
\begin{pmatrix}
0 & \mathbf{d}_1'\\
\mathbf{d}_1 & \operatorname{mat}(\mathbf{d}_2)
\end{pmatrix},
\qquad
\mathbf{d}=(\mathbf{d}_1',\mathbf{d}_2')'.
\]
A vector $\mathbf d$ is a recession direction of $\cM_{\le2}$ if and only
if
$M_0(\mathbf d)\succeq0$.
Sufficiency is immediate.
For necessity, divide
\[
M(\bdel)+tM_0(\mathbf d)\succeq0
\]
by $t$ and let $t\to\infty$; closedness of the positive-semidefinite cone
gives $M_0(\mathbf d)\succeq0$.

A positive-semidefinite matrix with a zero diagonal entry must have zero
off-diagonal entries in the corresponding row and column.
Therefore
$M_0(\mathbf d)\succeq0$ if and only if
\[
\mathbf d_1=\bzero,
\qquad
\operatorname{mat}(\mathbf d_2)\succeq0.
\]
Consequently,
\begin{equation}\label{eq:reccone}
\operatorname{rec}(\Theta_I)
=
\big\{(\bzero',\mathbf{d}_2')':\ \mathbf{d}_2\in\cN_2,\ \operatorname{mat}(\mathbf{d}_2)\succeq0\big\}.
\end{equation}

\emph{Characterization of \eqref{eq:reccone}.}
Under the convention \eqref{eq:Pc},
\[
P_{\mathbf d_2}(\bx)
=
\bx'\operatorname{mat}(\mathbf d_2)\bx,
\]
because $\binom{2}{\bb}$ equals one for square terms and two for cross
products.
By Theorem~\ref{thm:support},
$\mathbf d_2\in\cN_2$ if and only if
\[
\bx'\operatorname{mat}(\mathbf d_2)\bx=0
\qquad
\text{for every }\bx\in\cX.
\]
If
$\operatorname{mat}(\mathbf d_2)\succeq0$, this is equivalent to
\[
\|\operatorname{mat}(\mathbf d_2)^{1/2}\bx\|^2=0
\]
for every $\bx\in\cX$.
Hence
$\operatorname{mat}(\mathbf d_2)\bx=\bzero$ on
$\operatorname{span}(\cX)$.
If
$\operatorname{span}(\cX)=\mathbb R^q$, then
$\operatorname{mat}(\mathbf d_2)=\bzero$, so
$\operatorname{rec}(\Theta_I)=\{\bzero\}$ and $\Theta_I$ is compact.

Conversely, suppose
$\operatorname{span}(\cX)\neq\mathbb R^q$ and take a nonzero
$\bv$ orthogonal to $\cX$.
Let $\mathbf d_2$ be the coefficient vector corresponding to
$\bv\bv'$ in the ordering of $\bdel^{[2]}$.
Then
\[
\bx'\bv\bv'\bx=0
\]
for every $\bx\in\cX$ and
$\bv\bv'\succeq0$, so
$(\bzero',\mathbf d_2')'$ is a nonzero recession direction.
Along
\[
\bdel_0^{[\le2]}
+
t(\bzero',\mathbf d_2')',
\]
the candidate covariance matrix is
\[
\bSig_0+t\bv\bv',
\]
and
\[
\bv'(\bSig_0+t\bv\bv')\bv
=
\bv'\bSig_0\bv+t\|\bv\|^4
\longrightarrow\infty.
\]

\emph{Equivalent order-one conditions.}
By Theorem~\ref{thm:support} at $r=1$, $\cN_1$ is the space of linear
forms vanishing on $\cX$, hence the orthogonal complement of
$\operatorname{span}(\cX)$.
Therefore
\[
\cN_1=\{\bzero\}
\quad\Longleftrightarrow\quad
\operatorname{span}(\cX)=\mathbb R^q.
\]
Since $\bR_1(\bW)=\bW$, the identity-weight Gram matrix is
\[
\bM_{1\bI}=\E[\bW'\bW].
\]
When $\E\|\bW\|^2<\infty$, the identity weight is admissible, and
Theorem~\ref{thm:support} implies that this matrix is nonsingular if and
only if the order-one support condition holds.

\emph{Attainment.}
A linear functional on a nonempty compact set attains its extrema.
Every extremum belongs to $\Theta_I$, and by
Theorem~\ref{thm:spectra} is the moment vector of a possibly singular
Gaussian distribution with covariance
\[
\operatorname{mat}(\bdel^{[2]})
-
\bdel^{[1]}(\bdel^{[1]})'
\succeq0.
\]
\end{proof}

\begin{proof}[Proof of Proposition~\ref{prop:binarybounds}]
For the first moments, the two support points generate the one-dimensional
row spaces spanned by $(1,0)'$ and $(1,1)'$.
Together these span $\mathbb R^2$, so
$\cN_1=\{0\}$ and
$\bdel^{[1]}=(\mu_0,\mu_1)'$ is point identified by
Theorem~\ref{thm:support} at $r=1$.

For the second moments, the rows of $\bR_2(w)$ at $x=0$ and $x=1$ are
the coefficient vectors of
$(D_0+xD_1)^2$ in the ordering
$(D_0^2,D_0D_1,D_1^2)$:
\[
(1,0,0)
\qquad\text{and}\qquad
(1,2,1).
\]
Their common null space is
\[
\operatorname{span}\{(0,1,-2)'\}.
\]
Parametrize the affine class by $t\in\mathbb R$.
Because the first moments are fixed, shifting the raw second moments by
$t(0,1,-2)'$ shifts the coefficient covariance to
\[
\Sigma(t)
=
\begin{pmatrix}
\sigma_{00} & \sigma_{01}+t\\
\sigma_{01}+t & \sigma_{11}-2t
\end{pmatrix},
\]
where $\sigma_{jk}$ denote the true covariances.

By Theorem~\ref{thm:spectra}, feasibility is
$\Sigma(t)\succeq0$, that is,
$\sigma_{00}\ge0$ and
\[
\det\Sigma(t)
=
-t^2-2(\sigma_{00}+\sigma_{01})t+\sigma_{00}\sigma_{11}-\sigma_{01}^2
\ge0
\]
with $\sigma_{11}-2t\ge0$ additionally when needed; for
$\sigma_{00}>0$ the latter follows from the determinant condition.
The roots are
\[
t_{\pm}
=
-(\sigma_{00}+\sigma_{01})
\pm
\sqrt{\sigma_{00}(\sigma_{00}+2\sigma_{01}+\sigma_{11})}
=
-(\sigma_{00}+\sigma_{01})\pm s_0s_1,
\]
using
\[
s_0^2=\Var(Y\mid X=0)=\sigma_{00}
\]
and
\[
s_1^2
=
\Var(Y\mid X=1)
=
\sigma_{00}+2\sigma_{01}+\sigma_{11}.
\]
Since
$\Var(D_1)=\sigma_{11}-2t$ along the affine class, it is decreasing in
$t$.  Its identified interval is therefore
\[
\big[\sigma_{11}-2t_+,\ \sigma_{11}-2t_-\big]
=
\big[s_0^2+s_1^2-2s_0s_1,\ s_0^2+s_1^2+2s_0s_1\big]
=
\big[(s_1-s_0)^2,\ (s_1+s_0)^2\big].
\]
Every point in the interval is attained by a Gaussian structure by
Theorem~\ref{thm:spectra}.
If $s_0=0$, positive semidefiniteness implies
$\sigma_{01}=0$ and feasibility forces $t=0$, giving the singleton
$s_1^2$, as in \eqref{eq:binaryinterval}.

For the covariance,
\[
\Cov(D_0,D_1)
=
\sigma_{01}+t
\in
[-\sigma_{00}-s_0s_1,\,
 -\sigma_{00}+s_0s_1].
\]
\end{proof}

\section{Additional support geometry and partial identification}
\label{app:geometry}

\subsection{Dimension of model-level identified linear functionals}

Section~\ref{sec:partial} characterizes the identified set when the support
condition fails.
For a fixed order $r$, two candidate moment vectors
$\bdel_1^{[r]}$ and $\bdel_2^{[r]}$ generate the same conditional moments
over the regressor support if and only if their difference belongs to
$\cN_r$.
The linear functionals identified as parameters of the maintained
conditional-moment model therefore form the annihilator of
$\cN_r$ and have dimension
\begin{equation}\label{eq:hilbert}
Q_r-\dim\cI_r(\cX).
\end{equation}
For the stacked moment vector through order $r$, the corresponding dimension
is
\begin{equation}\label{eq:hilbertstack}
Q_{\le r}-\dim\cN_{\le r}
=
\sum_{s=1}^r\left[Q_s-\dim\cI_s(\cX)\right].
\end{equation}
Thus the same spaces of homogeneous polynomials that determine whether point
identification fails also determine the number of linearly independent
moment combinations that remain identified as parameters of the maintained
model; at particular degenerate distributions, moment feasibility can
identify additional functionals (Remark~\ref{rem:fixedvsmodel}).

\subsection{Simple sufficient conditions}

Each of the following is sufficient for the order-$r$ support condition:
\begin{enumerate}[label=(\roman*)]
\item the marginal support of some fixed row of $\bW$ is not contained in
the zero set of any nonzero homogeneous polynomial of degree $r$;
\item $\cX$ contains a nonempty open subset of $\mathbb R^q$
(equivalently, because $\cX$ is closed under scalar multiplication, it
contains a nonempty open cone), in which case moments of every fixed order
are identified;
\item $\cX$ contains $Q_r$ points
$\bx_1,\ldots,\bx_{Q_r}$ for which the order-$r$ product evaluation matrix
$[\bx_j^{\bb}]_{j,\bb}$ is nonsingular.
\end{enumerate}
For (i), if a homogeneous polynomial vanishes on $\cX$, then it vanishes
on each row of every $w\in\cW$ and, by continuity, on the marginal support
of each row of $\bW$.
Thus any cross-sectional polynomial-support condition satisfied by one row
is sufficient for the panel support condition.

\subsection{Generic rank for finite collections of row spaces}

Suppose the regressor support is finite with row spaces
$L_1,\ldots,L_m$.
Point identification under the maintained conditional-moment restrictions
is equivalent to full column rank of
\[
\begin{pmatrix}
\bR_r(w^{(1)})\\
\vdots\\
\bR_r(w^{(m)})
\end{pmatrix}.
\]
For fixed $(q,T,r,m)$, the set of tuples
$(w^{(1)},\ldots,w^{(m)})$ for which this matrix has rank at least $k$ is
Zariski open: at least one $k\times k$ minor must be nonzero.
Consequently, if one configuration attains full column rank, then full
column rank holds for Lebesgue-almost-every configuration in the same
ambient parameterization.
Bodmann, Ehler, and Gr\"af (2018) use the same nonempty-Zariski-open
argument to obtain generic maximal-rank results for finite collections of
projection subspaces.

The design counterpart of these rank calculations is the choice of the row
spaces themselves.
Maximizing the minimum eigenvalue of the averaged Gram matrix in
Remark~\ref{rem:fixeddesign} is an $E$-optimal design problem (Ehrenfeld,
1955; Pukelsheim, 1993).
Bodmann, Ehler, and Gr\"af (2018) give a related explicit construction:
orthogonal projection subspaces satisfying Grassmannian cubature identities
yield closed-form moment reconstruction, with
$\sum_j\omega_jP_j=(T/q)\bI_q$ at first order, where $P_j$ is the
orthogonal projector onto the $j$th $T$-dimensional row space.
Condition~\eqref{eq:fixeddesign-rank} instead applies to an arbitrary
sequence of regressor histories and requires only that the minimum
eigenvalue of the averaged Gram matrix remain bounded away from zero.
The row spaces determine whether full rank is attainable; the eigenvalue
criterion additionally depends on the normalization of the regressor
histories and on the weights, since $\bR_r(Fw)=\bR_r(F)\bR_r(w)$ for
nonsingular $F$ changes the Gram matrix without changing the row space.

When $T=2$ and $q\ge4$, the row spaces correspond projectively to lines in
$\mathbb P^{q-1}$.
The Hartshorne--Hirschowitz theorem (Hartshorne and Hirschowitz, 1982) that a generic union of lines in
$\mathbb P^n$, $n\ge3$, has good postulation implies the generic rank
\[
\min\{m(r+1),Q_r\}.
\]
For more general collections of subspaces, Derksen (2007) provides
Hilbert-polynomial and Hilbert-series tools under transversal or
combinatorially specified intersection patterns.
These results are used only to describe generic rank calculations; none is
required for Theorem~\ref{thm:support} or for the hyperplane and
two-history results in the main text.

\subsection{Finite exact-arithmetic certificates}

For a finite regressor support, direct computation of the rank of the
stacked matrix is an exact identification certificate.
If the entries of the support matrices are rational, clear denominators so
that the stacked matrix has integer entries.
If its reduction modulo some prime $p$ has column rank $Q_r$, then the
original matrix has column rank $Q_r$ over $\mathbb Q$, because a nonzero
minor modulo $p$ corresponds to a nonzero integer minor.
This provides a finite and reproducible rank certificate.

Several low-dimensional cases illustrate the rank calculation:
\begin{center}
\small
\begin{tabular}{@{}>{\raggedright\arraybackslash}p{0.34\textwidth}cccc@{}}
\toprule
Regressor support & $(q,T,r)$ & $m$ & $Q_r$ & rank\\
\midrule
one full-rank history & $(3,3,2)$ & 1 & 6 & 6\\
planes in $\mathbb R^3$ & $(3,2,2)$ & 1,2,3 & 6 & 3,5,6\\
planes in $\mathbb R^3$ & $(3,2,3)$ & 2,3,4 & 10 & 7,9,10\\
hyperplanes in $\mathbb R^4$ & $(4,3,2)$ & 2,3 & 10 & 9,10\\
$T=1$ generic row spaces & $(3,1,2)$ & 5,6 & 6 & 5,6\\
binary scalar $X$, within-unit movement & $(2,2,2)$ & 1 full-rank history & 3 & 3\\
binary scalar $X$, no within-unit movement & $(2,2,2)$ & 2 rank-one histories & 3 & 2\\
\bottomrule
\end{tabular}
\end{center}
The $(3,2,2)$ case shows why the counting condition
$mS_r\ge Q_r$ is not sufficient:
two planes supply six rows, but their one-dimensional intersection leaves
only five independent restrictions.

\section{A note on primitive moment conditions}\label{app:moments}

The main text states the integrability conditions directly in terms of the
order-$r$ moment contribution because these conditions allow dependence
between $\bD$ and $\bW$.
For bounded weights, simple primitive sufficient conditions can be given
separately for consistency and for asymptotic normality.

Since
\[
\bY^{[r]}=\bR_r\bD^{[r]},
\]
the moment contribution
\[
\bpsi_r
=
\bR_r'\bA_r\bR_r
(\bD^{[r]}-\bdel^{[r]})
\]
contains two factors of $\bR_r(\bW)$ and one factor of
$\bD^{[r]}$.
Because
\[
\|\bR_r\|\lesssim\|\bW\|^r,
\]
boundedness of $\bA_r$ implies, up to constants depending on
$\bdel^{[r]}$,
\[
\|\bpsi_r\|
\lesssim
\|\bW\|^{2r}(1+\|\bD\|^r).
\]
Hence
\begin{equation}\label{eq:primitive-cons}
\E\left[\|\bW\|^{2r}\big(1+\|\bD\|^{r}\big)\right]<\infty
\end{equation}
implies \eqref{eq:int1} and
$\E\|\bpsi_r\|<\infty$, which are the moment conditions used for
consistency in Theorem~\ref{thm:asynorm}.

Similarly,
\begin{equation}\label{eq:primitive}
\E\left[\|\bW\|^{4r}\big(1+\|\bD\|^{2r}\big)\right]<\infty
\end{equation}
implies \eqref{eq:int2}--\eqref{eq:int3}, because
\[
\|\bR_r'\bA_r\bR_r\|^2
\lesssim
\|\bW\|^{4r}
\]
and
\[
\|\bpsi_r\|^2
\lesssim
\|\bW\|^{4r}(1+\|\bD\|^{2r}).
\]
Thus the sufficient moment order required of $\bW$ is twice that required
of $\bD$ at each of these two stages.
Conditions of the symmetric form
\[
\E[\|\bD\|^{2\kappa r}+\|\bW\|^{2\kappa r}]<\infty,
\qquad
\kappa=1,2,
\]
would therefore impose a higher moment order on $\bD$ than
\eqref{eq:primitive-cons}--\eqref{eq:primitive} require; as explained
next, under Assumption~\ref{ass:moment} alone they are also not sufficient
substitutes, because marginal moment conditions do not control the required
cross moments.

Under full independence of $\bD$ and $\bW$, the joint moments factorize:
\eqref{eq:primitive-cons} follows from
\[
\E\|\bW\|^{2r}+\E\|\bD\|^r<\infty,
\]
and \eqref{eq:primitive} follows from
\[
\E\|\bW\|^{4r}+\E\|\bD\|^{2r}<\infty.
\]
Under Assumption~\ref{ass:moment} alone, these marginal moment conditions
are not sufficient.
Moment homogeneity restricts the conditional mean of
$\bD^{[r]}$ but not the relevant conditional dispersion, so it does not
control cross moments such as
\[
\E[\|\bW\|^{2r}\|\bD\|^r].
\]
For unbounded weights, including the oracle weight $\bGam_r^{+}$, the
corresponding weighted expectations are therefore best imposed directly.

\end{document}